\documentclass[11pt,letterpaper]{article}

\usepackage[margin=1in]{geometry}

\usepackage{amsfonts, amsmath, amssymb, amsthm}
\usepackage{mathrsfs} % \mathscr{P} or \mathscr{B}
\usepackage{setspace}

\usepackage{hyperref}
\hypersetup{
    colorlinks=true,
    linkcolor=teal,
    citecolor=teal,
    urlcolor=purple,
}
\usepackage{graphicx}           
\usepackage{xcolor}
\usepackage{algorithm}
\usepackage{algorithmic}
\usepackage{subfig}
\usepackage{booktabs}  

\newtheorem{theorem}{Theorem} % \newtheorem{theorem}{Theorem}[section] for numbering by sections		
\newtheorem{lemma}{Lemma} % \newtheorem{lemma}[theorem]{Lemma} to follow Theorem numbering

\theoremstyle{definition}
\newtheorem{definition}{Definition}	
\newtheorem{remark}{Remark}

\newcommand{\cA}{\mathcal{A}}

\newcommand{\cC}{\mathcal{C}}
\newcommand{\cD}{\mathcal{D}}

\newcommand{\cI}{\mathcal{I}}

\newcommand{\cS}{\mathcal{S}}

\newcommand{\cX}{\mathcal{X}}
\newcommand{\cY}{\mathcal{Y}}
\newcommand{\cZ}{\mathcal{Z}}

\newcommand{\R}{\mathbb{R}}

\newcommand{\E}{\mathbb{E}} % no spacing
\renewcommand{\P}{\mathbb{P}}

\DeclareMathOperator*{\argmin}{arg\,min} % \, for spacing
\DeclareMathOperator*{\argmax}{arg\,max} % \, for spacing

\renewcommand{\epsilon}{\varepsilon}

\title{Weighted Conformal Clustering}
\author{Anirban Nath \\
    Department of Statistics, Columbia University\\
    and \\
    YoonHaeng Hur\footnote{Corresponding author: yoonhaeng.hur@columbia.edu} \\
    Department of Statistics, Columbia University\\
    and \\
    Genevera I. Allen \\
    Department of Statistics, Columbia University}

\begin{document}

\maketitle

\begin{abstract}
    Clustering is a central tool for discovering latent structure in unlabeled data; yet modern clustering pipelines often end with a hard assignment of each observation to a cluster without rigorous measures of assignment uncertainty. We propose a novel weighted conformal approach for constructing valid confidence sets for cluster labels. The key difficulty is that the labels available for calibration are not observed ground-truth labels, but synthetic labels produced by a data-dependent clustering algorithm. Our method develops a conformal inference algorithm that corrects the resulting mismatch with the latent target labels through weights by formulating conformal clustering as a conditional label-distribution shift problem. We first derive an oracle procedure that attains finite-sample marginal coverage and then develop a computationally tractable and implementable version using estimated conditional label probabilities and novel augmented calibration. We show that the coverage of the estimated-weight procedure depends on the estimator, giving an explicit bound on the loss relative to the nominal level. Empirical studies demonstrate that the proposed weighted approach offers improvements over the recently proposed split conformal clustering procedure in terms of informative confidence set size, especially in nonlinear and high-dimensional clustering applications.
\end{abstract}

\section{Introduction}
\label{sec:intro}
Clustering is a central tool deployed to uncover latent group structures, which has been integrated as a fundamental component of modern data analysis pipelines across many scientific and industrial domains \cite{kiselev2019challenges, fortunato2010community}. Despite the advances in methodology and software, clustering often culminates in a hard assignment of each observation to a cluster without rigorous measures of assignment uncertainty, which can be problematic when the resulting cluster labels are used for subsequent scientific interpretation or operational decisions. Accordingly, it is instrumental to develop rigorous methods for uncertainty quantification for clustering, which has been approached from macroscopic or structural perspectives in the literature. Prominent methodologies include assessing centroid variability via bootstrapping \cite{kerr2001bootcluster, hofmans2015bootstrapkmeans, liu2018ukmeans}, testing the global existence and separation of clusters \cite{liu2008statistical, huang2015statistical, kimes2017statistical, shen2024statistical, yun2023selective, chen2023selective, gao2024selective}, and generating Bayesian credible sets over entire partitions \cite{wade2018bayesian, dahl2022search}. However, these global structural measures largely overlook a crucial, more granular target: the uncertainty inherent in the specific cluster assignment of an individual observation. 

Quantifying cluster assignment uncertainty is uniquely challenging because cluster labels are categorical, invariant to permutations, and typically generated by deterministic procedures deeply dependent on the observed dataset. Hence, these assignments lack the properties of typical continuous parameters, making the standard error-based analysis inapplicable. A further difficulty is that the partition returned on the observed data does not, by itself, define a calibrated rule for labeling new points or for expressing uncertainty about those labels. To address this, we take a viewpoint that frames label uncertainty as the reliability of a predictive labeling function and introduce a versatile methodology that transforms any base clustering algorithm into a system capable of generating statistically valid confidence sets for cluster membership anywhere in the feature space. These sets offer highly interpretable uncertainty quantification: singletons indicate clear, unambiguous assignments, while multi-label sets highlight ambiguous regions where several cluster identities are possible.

Recently, \cite{hur2026inference} proposed a conformal-inference approach to quantify uncertainty in cluster labels. Conformal inference is a general framework for constructing prediction sets with finite-sample, distribution-free coverage guarantees under exchangeability, without requiring a correctly specified model for the data distribution \cite{vovk2005algorithmic}. This framework has been especially influential in supervised learning, where it has been used to quantify predictive uncertainty in regression \cite{lei2018distribution,romano2019conformalized}, classification \cite{lei2014classification,sadinle2019least,romano2020classification}, and other applications. Its appeal lies in the fact that conformal methods can be applied on top of many predictive algorithms, converting fitted scores or predictions into set-valued outputs with rigorous coverage guarantees using a suitable calibration sample. \cite{hur2026inference} demonstrated that conformal inference can be adapted to unsupervised clustering in a principled way by framing label assignment as a predictive problem.\footnote{As noted in \cite{hur2026inference}, their framework is different from related approaches under the name of ``conformal clustering'' \cite{cherubin2015conformal,nouretdinov2020multi,kiani2020conformalized}, which use conformal p-values \cite{bates2023testing,liang2024integrative,lee2025full} to characterize and cluster outliers, rather than quantifying uncertainty.} At a high level, their approach adapts split conformal classification to clustering by using stochastic cluster labels to capture uncertainty, fitting a predictive rule for these labels, and calibrating conformity scores to return confidence sets for cluster membership at new points. However, \cite{hur2026inference} specifically states that since the calibration labels are synthetically generated by the clustering algorithm rather than sampled from a true ground-truth distribution, this procedure fundamentally violates the exchangeability assumption required for standard conformal validity. The paper then quantifies the resulting undercoverage in terms of the consistency and replace-one stability of the clustering algorithm, with a particular focus on smooth parametric mixture models. While this framework provides a valuable first step towards conformal clustering, important theoretical and methodological questions remain, with the need for investigation beyond the mixture-model setting and for more general clustering algorithms. 

In this paper, we develop a novel conformal inference framework for uncertainty quantification of cluster labels by utilizing weighted conformal prediction to address the distributional mismatch. Introduced by \cite{tibshirani2019conformal} in the context of covariate shift, weighted conformal prediction extends conformal inference beyond iid settings by reweighting the calibration distribution with likelihood-ratio weights, and has since motivated a broader line of work on conformal prediction under distribution shift. This perspective is well-suited to clustering because the labels used for calibration are generated by a clustering algorithm, whereas the desired coverage is with respect to the unobserved latent labels. We address this first by theoretically capturing the oracle conformal weights, which guarantee exact finite-sample coverage when the true conditional label distribution is known. Then, we develop a computationally tractable and fully data-driven algorithm that replaces the oracle quantities with estimated conditional probabilities and show that the coverage loss due to estimation is explicitly bounded by the estimation error of the chosen conditional label estimator. Finally, we empirically demonstrate the advantages of this method over \cite{hur2026inference} in different settings.

\paragraph{Our Contribution} We make three main contributions. First, we introduce a novel framework for uncertainty quantification of cluster labels that formulates this as a conditional label-distribution shift problem. We use a weighted conformal framework beyond exchangeable settings to quantify the validity of this procedure. Secondly, we propose a new augmented calibration mechanism that enjoys the benefits of both full-conformal and split-conformal approaches. This is an innovative idea with potential applications in conformal prediction for unsupervised learning beyond clustering, which broadens the scope of weighted conformal approaches in general. Finally, we show that this new framework achieves valid and more informative confidence sets in high-dimensional and nonlinear clustering problems, including spectral clustering applications, thereby extending the conformal clustering framework of \cite{hur2026inference} beyond its original setting.

\section{Weighted Conformal Clustering: Coverage Guarantees}
\subsection{Problem Formulation}
\label{sec:problem_formulation}
Adopting the inference framework for cluster labels proposed by \cite{hur2026inference}, we formally state the problem as follows. Given covariates \(X_1, \ldots, X_n \in \R^p\), clustering assigns labels in \([K] := \{1, \ldots, K\}\), where \(K\) is fixed and known. We assume a latent-label model in which \((X_i, Y_i^\ast)_{i = 1}^n\) are i.i.d. from an unknown distribution $P^\ast$ on $\R^p\times[K]$, but only the covariates are observed. For a new test point $(X_{n+1},Y_{n+1}^\ast)\sim P^\ast$, independent of the observed data, our goal is to construct a set-valued map $\hat{\cC} \colon \R^p \to 2^{[K]}$ that contains the true cluster label with high probability. Because cluster labels are identifiable only up to permutation, coverage must be stated after aligning the output labels with the latent ground-truth labels. To this end, the notion of oracle alignment was introduced in \cite{hur2026inference} as follows: $\hat{\sigma}_o^\ast \in \argmax_{\sigma\in\cS_K}
\P(\sigma(Y^\ast)\in \hat{\cC}(X)\mid \hat{\cC})$, where $\cS_K$ is the set of all permutations of $[K]$ and $(X,Y^\ast)\sim P^\ast$ is independent of the randomness used to construct $\hat{\cC}$. We seek a confidence set satisfying, for a prescribed level $\alpha\in(0,1)$,
\begin{equation}
    \label{eq:target_coverage_bound}
    \P\left(\hat{\sigma}_o^\ast(Y_{n+1}^\ast)\in \hat{\cC}(X_{n+1})\right)\ge 1-\alpha.
\end{equation}
Thus, $\hat{\cC}(x)$ quantifies pointwise cluster-label uncertainty: a singleton set indicates a relatively certain assignment, while a larger set indicates that multiple cluster labels remain plausible.

Since we do not know the true labels, we have to rely on the labels generated by the clustering algorithm. For that, we define the following notation. 

\begin{definition}[Stochastic clustering algorithm]
    \label{def:stochastic_clustering}
    Let $\mathcal{X}$ denote the domain on which cluster assignments are defined, and let $\Delta_K \subset \R^K$ be the probability simplex. A clustering algorithm $\mathcal{A}$ is called a \emph{stochastic clustering algorithm} if, for any input dataset $\cD_n = \{x_1,\dots,x_n\} \subseteq \R^p$, it outputs a mapping $\mathcal{A}_{\cD_n} : \mathcal{X} \to \Delta_K$. For each $x \in \mathcal{X}$, the vector $\mathcal{A}_{\cD_n}(x) = \bigl(\mathcal{A}_{\cD_n}(x)_1,\dots,\mathcal{A}_{\cD_n}(x)_K\bigr)$ represents the soft cluster assignment probabilities for $x$, where $[\mathcal{A}_{\cD_n}(x)]_k$ is interpreted as the probability that $x$ belongs to cluster $k$. A realized cluster label for $x$ is then generated according to $Y \mid x,\cD_n \sim \mathrm{Cat}\bigl(\mathcal{A}_{\cD_n}(x)\bigr)$ where, for any $w \in \Delta_K$, Cat($w$) denotes the categorical distribution on $[K]$ that draws $k \in [K]$ with probability $w_k$.
\end{definition}

\begin{remark}
    For this paper, we consider the map $\cD \mapsto \mathcal A_{\cD}$ that is permutation-invariant in its argument, in the sense that it depends only on the unordered set $\cD$.
\end{remark}

The framework in Definition~\ref{def:stochastic_clustering} is broad enough to encompass all clustering procedures essentially. For non-generalizable methods such as k-means or hierarchical clustering, assignments are produced only for the observed data points, so $\mathcal X=\cD_n$. In contrast, model-based methods such as Gaussian mixture models define assignment probabilities for any $x\in\R^p$, in which case $\mathcal X=\R^p$, and we call the procedure a \emph{generalizable stochastic clustering algorithm}. Moreover, deterministic clustering algorithms can also be explained as a special case: if an algorithm assigns hard labels $y_1,\ldots,y_n\in[K]$, then we may set $\left[\mathcal A_{\cD_n}(x_i)\right]_k= 1\{k=y_i\}$, so that $Y_i\sim\mathrm{Cat}(\mathcal A_{\cD_n}(x_i))$ is degenerate at the assigned label. Thus, hard clustering corresponds to point-mass assignment probabilities, while soft clustering allows nondegenerate probabilities over multiple clusters.

\subsection{Prelude: Oracle-Weighted Split Conformal Clustering and Its Limitations}
As discussed in Section~\ref{sec:intro}, the proposed uncertainty quantification framework of \cite{hur2026inference} through conformal prediction generates labels by a fitted clustering rule. Since the synthetic cluster labels are produced in a data-dependent manner, the cluster-labeled calibration sample is not exchangeable with the pair of test point and its true label, so the classical conformal validity argument does not apply directly. To make this precise, note that for any unordered covariate set \(\mathcal D\), the stochastic clustering algorithm \(\mathcal A_{\mathcal D}\) induces a conditional label distribution on \([K]\) via $Y \mid X=x,\mathcal D \sim \mathrm{Cat}(\mathcal A_{\mathcal D}(x))$, which along with the true covariate marginal \(P_X^\ast\), yields a synthetic joint distribution, say $\tilde P$ on \(\R^p \times [K]\), whereas the target distribution governing the true latent labels is $P^\ast$ defined below: 
\begin{equation*}
    \tilde P(dx,dy) := P_X^\ast(dx)\sum_{k=1}^K [\mathcal A_{\mathcal D}(x)]_k\,\delta_k(dy), \quad P^\ast(dx,dy) = P_X^\ast(dx)\sum_{k=1}^K p_{Y\mid X}^\ast(k\mid x)\,\delta_k(dy),
\end{equation*}
where $p_{Y\mid X}^\ast (\cdot\mid\cdot)$ is the conditional density corresponding to the oracle conditional distribution of the labels, namely $P_{Y\mid X}^\ast$. Thus, \(\tilde P\) and \(P^\ast\) have different joint laws although they share the same marginals. The clustering procedure supplies calibration labels from \(\tilde P\), whereas the desired guarantee in \eqref{eq:target_coverage_bound} concerns labels from \(P^\ast\). This perspective frames the problem as a form of conditional label-distribution shift, or \textit{concept shift}, where validity requires correcting the discrepancy between the synthetic label mechanism induced by \(\mathcal A_{\mathcal D}\) and the true latent-label mechanism.

Distribution shift has been widely studied in the conformal inference literature, with much of the focus on covariate shift \cite{tibshirani2019conformal} and label shift \cite{podkopaev2021distribution,si2023pac}, which can be extended to more general forms of shift. The key idea behind this is to reweight the calibration scores when forming the empirical score distribution. With correctly specified weights, the resulting weighted distribution matches the conditional empirical distribution of the test-point score, and its weighted quantile can be used to determine whether the test point is included in the prediction set. The standard conformal procedure is a special case of this, where we choose uniform weights for all points. However, many existing methods rely on the exchangeability of the augmented calibration-test sample, which is absent in our setting. The problem considered here is thus more delicate because we not only encounter a distribution shift problem but also a non-exchangeable data-generating procedure. Therefore, we must view the problem through a more general lens of conformal prediction under arbitrary joint distributions, where validity does not require exact exchangeability of the augmented sample. We briefly review this framework in Appendix~\ref{appendix:wcp_theory}.

In \cite{hur2026inference}, the failure of exchangeability is solely treated as a source of error whose effect can be quantified, leading to undercoverage bounds and asymptotic validity under additional stability and model assumptions. In contrast, we want to address this exchangeability issue more directly by seeking a conformal construction that accounts for the distributional mismatch induced by synthetic cluster labels. To this end, we attempt to derive the oracle-weighted version of the split conformal clustering \cite{hur2026inference}. Instead of assigning uniform weights to the conformity scores, we calculate the weights that correct the discrepancy discussed above, leading to a weighted empirical score distribution and a corresponding weighted prediction rule.

Based on the general theory of conformal prediction adapted to our setting, we can choose the desired weights as follows. 
Let \(\cD_{ca}=\{X_i\}_{i \in \cI_{ca}}\) denote the calibration covariates with $\cI_{ca}$ denoting their indices. Let \((X_i,Y_i)_{i \in \cI_{ca}}\) be the corresponding calibration sample, where the labels are generated by the stochastic clustering algorithm $\cA_{\cD_{ca}}$ as defined in Definition \ref{def:stochastic_clustering}. Let $x$ be the test covariate of interest. For $i \in \cI_{ca}$, let $\cD_{ca, \{x\}}^{-i} = \cD_{ca} \cup \{x\} \setminus \{X_i\}$ and define the unnormalized weights as
\begin{equation}
    \label{eq:weight_calibration} 
    \tilde w_i(y)
    \propto
    p_{Y\mid X}^\ast\left(Y_i \mid X_i\right)
    \left(
    \prod_{j \in \cI_{ca} \setminus \{i\}}
    \left[\mathcal A_{\cD_{ca, \{x\}}^{-i}}(X_j)\right]_{Y_j}
    \right)
    \left[\mathcal A_{\cD_{ca, \{x\}}^{-i}}(x)\right]_{y},
\end{equation}
while the test-point weight is
\begin{equation}
    \label{eq:weight_test}
    \tilde w(y)
    \propto
    p_{Y\mid X}^\ast(y\mid x)
    \prod_{j \in \cI_{ca}}
    \left[\mathcal A_{\cD_{ca}}(X_j)\right]_{Y_j},
\end{equation}
which are normalized to \(\{w_i(y)\}_{i \in \cI_{ca}} \cup \{w(y)\}\) that sum to one. Using these weights, we derive the oracle-weighted version of split conformal clustering proposed in \cite{hur2026inference}; see Algorithm~\ref{alg:oracle_weighted_conformal_clustering}.

\renewcommand{\thealgorithm}{0} % Restore normal numbering format
\renewcommand{\theHalgorithm}{0} % Unique hyperref anchors (algorithm.0)
\begin{algorithm}[!htbp]
    \caption{Weighted Split Conformal Clustering (Oracle)}
    \label{alg:oracle_weighted_conformal_clustering}
    \begin{algorithmic}[1]
        \REQUIRE Covariates \(X_1, \ldots, X_n \in \R^p\), user-specified level \(\alpha \in (0, 1)\), test point $X_{n+1} = x$.
        \STATE Split the covariates into a training set \(\{X_i\}_{i \in \cI_{tr}}\) and a calibration set \(\{X_i\}_{i \in \cI_{ca}}\).
        \STATE Apply clustering to the training data \(\{X_i\}_{i \in \cI_{tr}}\) and obtain corresponding labels \(\{Y_i\}_{i \in \cI_{tr}}\).
        \STATE Fit a soft classifier \(\hat{\pi} \colon \R^p \to \Delta_K\) on the cluster-labeled training data \(\{(X_i, Y_i)\}_{i \in \cI_{tr}}\).
        \STATE Apply clustering to the calibration data \(\cD_{ca}\) and obtain corresponding labels $\{Y_i\}_{i \in \cI_{ca}}$. 
        \STATE For each candidate label \(y \in [K]\), calculate the calibration weights $w_i(y)$ for $i \in \cI_{ca}$ and the test weight $w(y)$ obtained by normalizing the weights in \eqref{eq:weight_calibration}-\eqref{eq:weight_test}.
        \STATE Define suitable conformity scores $S_i = s((X_i, Y_i); \hat{\pi})$ for $i \in \cI_{ca}$ and construct
        \begin{equation*}
            \hat{\cC}(x) = \left\{y \in [K] : s((x, y); \hat{\pi}) \le \mathsf{Q}_{1 - \alpha}\left({\textstyle\sum_{i \in \cI_{ca}}} w_i(y) \, \delta_{S_i} + w(y) \delta_{\infty}\right)\right\}.
        \end{equation*}
        \vspace*{-10pt}
        \RETURN \(\hat{\cC}(x)\).
    \end{algorithmic}
\end{algorithm}

Algorithm~\ref{alg:oracle_weighted_conformal_clustering} modifies the split conformal clustering procedure (Algorithm 1 of \cite{hur2026inference}) by introducing the weights in Step 5 and consequently constructs $\hat{\cC}$ by calculating the weighted quantile $\mathsf{Q}_{1 - \alpha}$ of the conformity scores; see Definition~\ref{def:quantile}. Here, $s$ in Step 6 is any suitable conformity score function used in conformal classification. In Algorithm 1 of \cite{hur2026inference}, Step 5 was the alignment between the calibration labels and the training classifier, which is skipped in Algorithm~\ref{alg:oracle_weighted_conformal_clustering} for the reason we explain later; see also Remark~\ref{rmk:why_no_alignment}. Now, the following theorem formalizes the marginal coverage guarantee of Algorithm~\ref{alg:oracle_weighted_conformal_clustering}.

\begin{theorem}[Finite-sample coverage of oracle-weighted algorithm]
    \label{thm:exact_weights}
    Suppose $X_1,\ldots,X_n \in \R^p$ are i.i.d. from $P^\ast_X$. Let \(\hat{\cC}\) be the output of Algorithm~\ref{alg:oracle_weighted_conformal_clustering}. Then, for $(X_{n+1},Y_{n+1}^\ast)\sim P^\ast$ independent of $\{X_i\}_{i=1}^n$, we have
    \(
    \P(\hat{\sigma}_o^\ast(Y_{n+1}^\ast)\in \hat{\cC}(X_{n+1}))\ge 1-\alpha.
    \)
\end{theorem}

The proof of Theorem~\ref{thm:exact_weights} is deferred to Appendix \ref{sec:proofs}. The idea is to utilize the general weighted conformal inference theory based on permutations, which we review in Appendix \ref{appendix:wcp_theory}. The crux of this theory is to derive permutation-dependent weights from the joint distribution. We show that these weights can be reduced to \eqref{eq:weight_calibration} and \eqref{eq:weight_test} in our clustering setting. Then, from the general theory, we can show the desired coverage. Lastly, we remark that while the aforementioned alignment step in Algorithm 1 of \cite{hur2026inference} is natural in practice, we have excluded this step in Algorithm \ref{alg:oracle_weighted_conformal_clustering} to keep the theoretical derivation straightforward and streamlined.

While Theorem~\ref{thm:exact_weights} provides the desired coverage guarantee, the exact oracle weights reveal two obstacles to implementation. First, they involve the unknown conditional density \(p_{Y\mid X}^\ast\), which is an oracle quantity and must be estimated in practice. Second, their sample-dependent terms require evaluating the clustering rule on the leave-one-out augmented covariate set \(\cD_{ca, \{x\}}^{-i}\) for \(i \in \cI_{ca}\). Hence, for each test covariate \(x\), the exact construction requires one clustering fit on \(\cD_{ca}\) and \(|\cI_{ca}|\) additional leave-one-out fits, which is computationally intractable. Also, this inevitably introduces label misalignment from multiple clustering fits. Therefore, beyond the statistical challenge of estimating the oracle label law, the exact oracle-weighted construction is computationally expensive even at the level of the sample-dependent terms, motivating a more tractable approximation.

\subsection{A Tractable Estimated-Weight Construction}
As the preceding discussion makes clear, a feasible weighted-conformal clustering procedure must not only overcome the computational burden of repeated clustering refits but also come with a tractable estimation of the latent conditional law $P_{Y\mid X}^\ast$. To alleviate the computational difficulty, we introduce a novel augmented calibration mechanism that eliminates the leave-one-out refitting burden in the oracle weights. Our key idea is to remove this dependence by modifying the synthetic-label generation mechanism itself. Instead of generating calibration labels from a clustering rule fitted on \(\cD_{ca}\), we fit the clustering algorithm once on the augmented covariate set for $X_{n + 1} = x$, namely,
\[
    \cD_{aug}:=\cD_{ca}\cup\{x\}.
\]
This augmentation is a simple yet novel idea that bridges the full-conformal and split-conformal prediction, utilizing the advantages of both mechanisms. Consequently, the joint synthetic likelihood factors that previously necessitated repeated refits become common multiplicative terms and cancel after normalization. The corresponding oracle weights reduce to the pointwise likelihood-ratio form as follows 
\begin{equation}
    \label{eq:augmented_oracle_weights}
    \tilde w_i
    \propto
    \frac{p_{Y\mid X}^\ast\left(Y_i\mid X_i\right)}
    {\left[\mathcal A_{\cD_{aug}}(X_i)\right]_{Y_i}}
    \quad \text{for} ~ i \in \cI_{ca}, 
    \qquad
    \tilde w(y)
    \propto
    \frac{p_{Y\mid X}^\ast(y\mid x)}
    {\left[\mathcal A_{\cD_{aug}}(x)\right]_y}.
\end{equation}
This transformation is both statistically and computationally consequential. Statistically, the weights retain the familiar distribution-shift interpretation, with each ratio comparing the target latent-label probability in the numerator to the synthetic label probability induced by the augmented clustering fit in the denominator. Computationally, the augmented construction replaces the \(|\cI_{ca}|+1\) clustering fits required by the exact oracle weights in \eqref{eq:weight_calibration}--\eqref{eq:weight_test} with a single fit on \(\cD_{aug}\). Thus, the augmented calibration step is the pivotal mechanism that turns the weighted conformal correction into a tractable framework for cluster-label uncertainty quantification.
 
The augmented construction removes the main computational obstacle, but the weights in \eqref{eq:augmented_oracle_weights} still depend on the unknown conditional density \(p_{Y\mid X}^\ast\). To obtain an implementable method, let \(\hat p_{Y\mid X}\) be a suitable estimator of the true conditional label density $p_{Y \mid X}^\ast$ based on the training split, which is aligned with the training clustering. We then replace the oracle numerator in \eqref{eq:augmented_oracle_weights} by \(\hat p_{Y\mid X}\), giving the estimated weights 
\begin{equation}
    \label{eq:estimated_augmented_weights}
    \hat w_i
    \propto
    \frac{\hat p_{Y\mid X}(\hat{\sigma}(Y_i)\mid X_i)}
    {\left[\mathcal A_{\cD_{aug}}(X_i)\right]_{Y_i}}
    \quad \text{for} ~ i \in \cI_{ca},
    \qquad
    \hat w(y)
    \propto
    \frac{\hat p_{Y\mid X}(y \mid x)}
    {\left[\mathcal A_{\cD_{aug}}(x)\right]_{\hat{\sigma}^{-1}(y)}},
\end{equation}
where $\hat{\sigma}$ is a permutation that aligns the augmented cluster labels to the training clustering and hence to $\hat{p}_{Y \mid X}$. The weights in \eqref{eq:estimated_augmented_weights} are normalized to \(\{w_i(y)\}_{i \in \cI_{ca}} \cup \{w(y)\}\) so that they sum to one, which are used in the weighted conformal distribution separately for each candidate label \(y\); here, we place $y$-dependence in the normalized weights because the normalization turns the unnormalized calibration weights $\hat{w}_i$'s in \eqref{eq:estimated_augmented_weights} into $y$-dependent weights. Algorithm~\ref{alg:estimated_weighted_conformal_clustering} summarizes the above procedure.

\setcounter{algorithm}{0} % Reset algorithm counter to 0
\renewcommand{\thealgorithm}{\arabic{algorithm}} % Restore normal numbering format
\renewcommand{\theHalgorithm}{\arabic{algorithm}} % Restore hyperref anchors (algorithm.x)
\begin{algorithm}[!htbp]
    \caption{Weighted Conformal Clustering (Augmented)}
    \label{alg:estimated_weighted_conformal_clustering}
    \begin{algorithmic}[1]
        \REQUIRE Covariates \(X_1, \ldots, X_n \in \R^p\), user-specified level \(\alpha \in (0, 1)\), test point $X_{n+1} = x$.
        \STATE Split the covariates into a training set \(\{X_i\}_{i \in \cI_{tr}}\) and a calibration set \(\{X_i\}_{i \in \cI_{ca}}\).
        \STATE Apply clustering to the training data \(\{X_i\}_{i \in \cI_{tr}}\) and obtain corresponding labels \(\{Y_i\}_{i \in \cI_{tr}}\).
        \STATE Fit a soft classifier \(\hat{\pi} \colon \R^p \to \Delta_K\) on the cluster-labeled training data \(\{(X_i, Y_i)\}_{i \in \cI_{tr}}\).
        \STATE Apply clustering to the augmented data \(\cD_{aug} := \cD_{ca} \cup \{x\}\) and obtain the labels $\{Y_i\}_{i \in \cI_{ca}}$ for the calibration data.
        \STATE Align the cluster and classification labels: for a suitable metric $\ell$ on $\Delta_K$, find 
        \begin{equation*}
            \hat{\sigma} \in {\textstyle\argmin_{\sigma \in \cS_K}} ~ {\textstyle \sum_{i \in \cI_{ca} \cup \{n + 1\}}} \ell(\sigma(\hat{\pi}(X_i)), \cA_{\cD_{aug}}(X_i)), \vspace*{-4pt}
        \end{equation*}
        where $\sigma(w) = (w_{\sigma(1)}, \ldots, w_{\sigma(K)})$ for any $w \in \Delta_K$ and $\sigma \in \cS_K$.
        \STATE Estimate a suitable conditional density $\hat p_{Y\mid X}$ based on the training split.
        \STATE For each candidate label \(y \in [K]\), calculate the calibration weights $w_i(y)$ for $i\in \cI_{ca}$ and the test weight $w(y)$ obtained by normalizing the weights in \eqref{eq:estimated_augmented_weights}.
        \STATE Define suitable conformity scores $S_i = s((X_i, \hat{\sigma}(Y_i)); \hat{\pi})$ for $i \in \cI_{ca}$ and construct
        \begin{equation*}
            \hat{\cC}(x) = \left\{y \in [K]: s((x, y); \hat{\pi}) \le \mathsf{Q}_{1 - \alpha}\left({\textstyle\sum_{i \in \cI_{ca}}} w_i(y) \, \delta_{S_i} + w(y) \delta_{\infty}\right)\right\}.
        \end{equation*}
        \vspace*{-10pt}
        \RETURN \(\hat{\cC}(x)\).
    \end{algorithmic}
\end{algorithm}

\begin{remark}
    In Algorithm~\ref{alg:estimated_weighted_conformal_clustering}, while $\hat{p}_{Y \mid X}$ can be any suitable estimator of $p_{Y \mid X}^\ast$ as long as it is independent of the calibration split, the most natural choice is to use the soft classifier $\hat{\pi}$ from Steps 2-3, namely, $\hat{p}_{Y \mid X}(\cdot \mid x) = \hat{\pi}(x)$. Also, note that Algorithm~\ref{alg:estimated_weighted_conformal_clustering} can accommodate hard clustering such as k-means. In this case, the augmented clustering fit produces degenerate assignment probabilities, with \([\mathcal A_{\cD_{aug}}(X_i)]_{Y_i}=1\), and $\hat{p}_{Y \mid X}$ can still be obtained by using the soft classifier $\hat{\pi}$ as explained earlier. Meanwhile, if one uses a generalizable soft clustering algorithm, the soft classifier $\hat{\pi}$ in Step 3 may be skipped since we can use it as $\hat{p}_{Y \mid X}$ directly. Lastly, note that the alignment in Step 5 is based on the augmented covariates and thus is different from the alignment in Step 5 of Algorithm 1 in \cite{hur2026inference}; also, we only use the covariates, not the obtained cluster labels. The alignment step can be done by the standard linear assignment algorithms in polynomial time. Lastly, we note that the denominators of the weights in \eqref{eq:estimated_augmented_weights} can be small or even zero (when using hard clustering), which can cause numerical instability. To deal with this issue, we clip the denominator, namely, set it as $\epsilon = 10^{-6}$ if it is below $\epsilon$. If both the numerator and the denominator are below $\epsilon$, so that the ratio is effectively $0/0$, we set the weight to $1$. The weights are then normalized as before.
\end{remark}

The following theorem derives a finite-sample bound on the coverage.

\begin{theorem}[Finite-sample coverage with estimated weights]
    \label{thm:coverage_estimated_weights}
    Suppose $X_1,\ldots,X_n \in \R^p$ are i.i.d. from $P^\ast_X$. Let \(\hat{\cC}\) be the output of Algorithm~\ref{alg:estimated_weighted_conformal_clustering}, where the conditional label distribution corresponding to the density $\hat{p}_{Y \mid X}$ in \eqref{eq:estimated_augmented_weights} is denoted as $\widehat{P}_{Y \mid X}$. Then, for $(X_{n+1},Y_{n+1}^\ast)\sim P^\ast$ independent of $\{X_i\}_{i=1}^n$, we have
    \[
        \P\left(\hat{\sigma}_o^\ast(Y_{n+1}^\ast)\in \hat{\cC}(X_{n+1})\right)
        \ge 1-\alpha - \min_{\sigma \in \cS_K} \E_{tr} \E_{X \sim P^\ast_X}\left[d_{\mathrm{TV}}\!\left(\widehat P_{Y\mid X}(\cdot\mid X), \sigma \circ P_{Y\mid X}^\ast(\cdot\mid X)\right)\right],
    \]
    where $d_{\mathrm{TV}}$ is the total variation distance over the simplex, $\E_{tr}$ is the expectation with respect to the training split, and $\E_{X \sim P_X^\ast}$ is the expectation with respect to $X \sim P_X^\ast$ independent of everything else. Here, $\sigma \circ P_{Y\mid X}^\ast(\cdot\mid X)$ denotes the label distribution $P_{Y\mid X}^\ast(\cdot\mid X)$ permuted according to $\sigma \in \cS_K$.
\end{theorem}

The proof of Theorem~\ref{thm:coverage_estimated_weights} is deferred to Appendix \ref{sec:proofs}. Theorem~\ref{thm:coverage_estimated_weights} combines two ideas: the augmentation mechanism of the calibration set and the estimation of the conditional density from the training set. Particularly, if we can perfectly estimate the oracle conditional density, namely, $\hat{p}_{Y \mid X} = p^\ast_{Y \mid X}$, then Algorithm~\ref{alg:estimated_weighted_conformal_clustering} has $1 - \alpha$ coverage. For general $\hat{p}_{Y \mid X}$, Theorem~\ref{thm:coverage_estimated_weights} shows that the coverage loss of the estimated-weight procedure is controlled by the error in estimating the latent conditional label probabilities. The proof again builds on the general weighted conformal inference theory that derives the weights from the joint density of the calibration and test points, which, unlike Theorem~\ref{thm:exact_weights}, leads to the simpler weights in \eqref{eq:estimated_augmented_weights} thanks to the augmentation.

\begin{remark}
    Notice that the optimization in Step 5 of Algorithm \ref{alg:estimated_weighted_conformal_clustering} may have multiple minimizers. In this case, we can use any deterministic, observation-order-invariant tie-breaking rule to define $\hat{\sigma}$. For example, one may select the lexicographically smallest permutation among all minimizers. Even a random tie-breaking rule is acceptable as long as the randomness is independent of everything else. Then, we can represent the permutation $\hat{\sigma}$ as a deterministic symmetric function of the augmented covariates, which allows us to use the general weighted conformal inference theory. We will formally define this in the proof of Theorem \ref{thm:coverage_estimated_weights}.
\end{remark}

Unlike \cite{hur2026inference}, our framework does not introduce an additional replace-one stability term, since the augmented calibration construction removes the instability from repeated leave-one-out clustering refits. The remaining requirement is consistency of the estimated conditional label law, which is intrinsic when targeting unobserved latent labels in an unsupervised setting. However, the two guarantees are not directly comparable, as our bound concerns out-of-sample conditional estimation error, while \cite{hur2026inference} studies the in-sample behavior of stochastic split conformal clustering. Nevertheless, under correctly specified smooth parametric mixtures, standard identifiability and regularity conditions yield consistent conditional label estimation, so both methods attain asymptotic validity. Beyond these classical mixture models, our weighted framework can accommodate more flexible conditional probability estimators, enabling broader clustering regimes; the next section illustrates this advantage in nonlinear examples, where our method maintains coverage with more informative prediction sets. Of course, for such nonlinear clustering examples, little is known about the consistency in general. Deriving conditions under which the clustering algorithm is consistent has been a long-standing problem in the literature, which is beyond the scope of this current paper. We leave the study of consistency of clustering that fits within our theoretical framework and \cite{hur2026inference} as future work.

While Theorem~\ref{thm:coverage_estimated_weights} emphasizes consistent estimation, exact model specification is not required: under misspecification, the coverage loss is governed by the discrepancy between the limiting estimated conditional law and the oracle law, which can be offset by a conservative choice of the nominal level. Although the framework can be applied to hard clustering methods such as k-means, their degenerate one-hot assignments cannot consistently recover a smooth latent conditional label distribution (see Remark~2 of \cite{hur2026inference}). So the undercoverage term in Theorem~\ref{thm:coverage_estimated_weights} generally does not necessarily vanish when using strictly deterministic assignments. Therefore, to fully capture assignment uncertainty and fully realize the theoretical guarantees afforded by this framework, we recommend using stochastic soft-clustering algorithms.

\section{Empirical Studies}

\subsection{Simulations}

\begin{figure}[t]
    \centering
    \includegraphics[width=\linewidth]{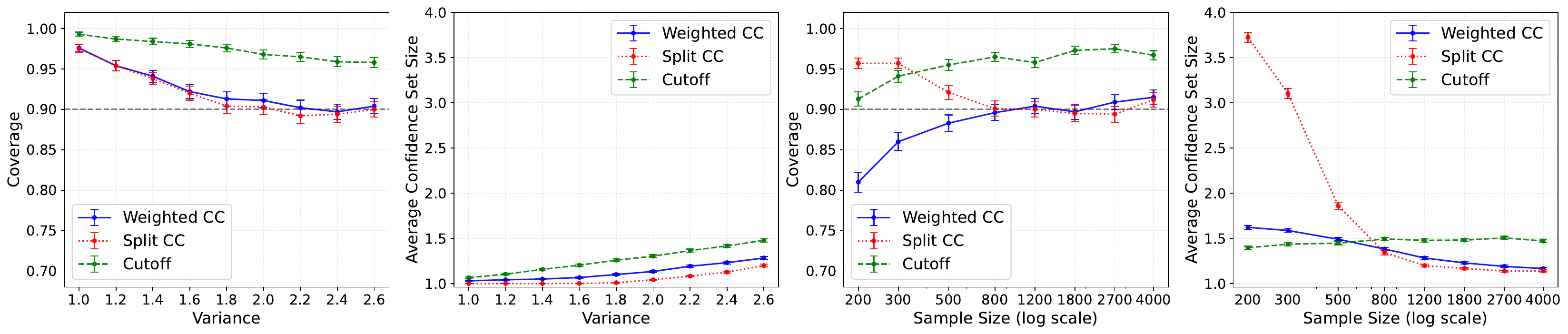}
    \includegraphics[width=\linewidth]{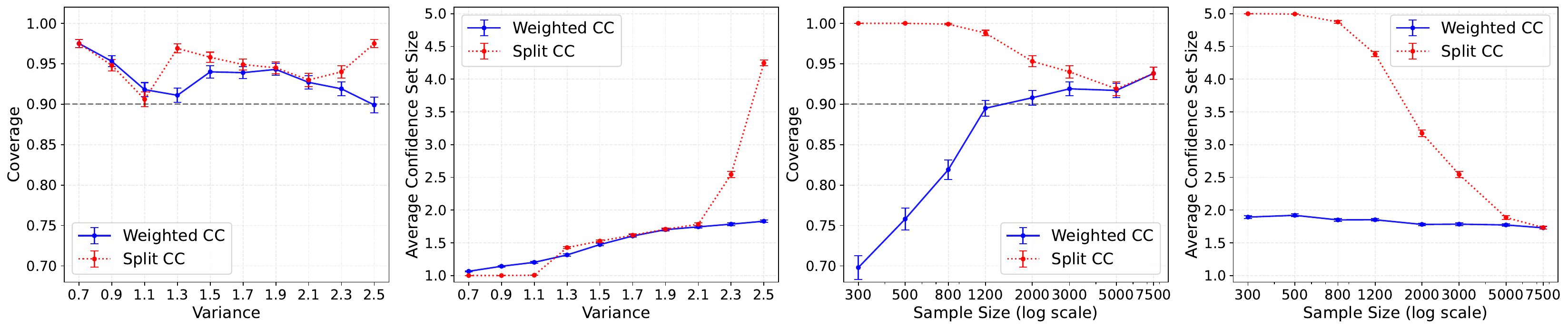}
    \vspace*{-10pt}
    \caption{Coverage and confidence set size for GMM simulations with $K = 5$ components for varying noise level $\sigma^2$ and sample size \(n\). Top row: GMM in 2D with stochastic GMM clustering. Bottom row: GMM in 50D with spectral clustering that applies stochastic GMM clustering to the top \(K\) singular vectors. Algorithm \ref{alg:estimated_weighted_conformal_clustering} is compared with the split conformal clustering (CC) and naive cutoff methods.}
    \label{fig:GMM}
\end{figure}

\paragraph{Gaussian Mixture Models} We simulate data from a \(K = 5\) component Gaussian mixture model (GMM) with equal mixing proportions in low-dimensional (\(\R^2\)) and high-dimensional (\(\R^{50}\)) settings. Mixture centers are fixed, sharing the common covariance matrix $\sigma^2 I_p$. For \(\alpha = 0.1\), we evaluate the coverage and set size of Algorithm \ref{alg:estimated_weighted_conformal_clustering} for varying $\sigma^2$ with fixed \(n = 1200\) (\(n = 3000\) for \(\R^{50}\)) and for varying sample size with fixed \(\sigma^2 = 2.6\) (\(\sigma^2 = 2.3\) for \(\R^{50}\)). Algorithm \ref{alg:estimated_weighted_conformal_clustering} is compared with the split conformal clustering (CC) method \cite{hur2026inference} and the naive cutoff method, which cuts the estimated soft labels to $1 - \alpha$. For \(\R^2\), the stochastic version of GMM clustering is used for both Algorithm \ref{alg:estimated_weighted_conformal_clustering} and split CC. For \(\R^{50}\), we use spectral clustering, which first applies singular value decomposition to the data and then applies stochastic clustering to the top \(K\) singular vectors. We use k-nearest neighbor classification (with k \(= \sqrt{n_{tr}}\)) and generalized inverse quantile \cite{romano2020classification} as the conformity score.

The top row of Figure \ref{fig:GMM} shows the results in \(\R^2\), plotting the coverage and set size as a function of \(\sigma^2\) or the sample size \(n\). We can see that both Algorithm \ref{alg:estimated_weighted_conformal_clustering} and split CC achieve valid coverage with informative set size. Meanwhile, as noted in \cite{hur2026inference}, the naive cutoff method results in larger sets because, by construction, it has to include less probable labels to achieve the target coverage. The bottom row shows the results in \(\R^{50}\) with spectral clustering. The first two plots show that Algorithm \ref{alg:estimated_weighted_conformal_clustering} remains valid and informative for the largest variances $\sigma^2 \in \{2.3, 2.5\}$, while split CC becomes oversized and hence uninformative. The last two plots show that Algorithm \ref{alg:estimated_weighted_conformal_clustering} starts to achieve valid coverage with informative set size at \(n = 1200\), while the set size of split CC remains large even at \(n = 3000\) and starts to match that of Algorithm \ref{alg:estimated_weighted_conformal_clustering} at \(n = 5000\). These results illustrate that Algorithm \ref{alg:estimated_weighted_conformal_clustering} can provide more informative uncertainty quantification for spectral clustering in high-dimensional settings.

\begin{figure}[t]
    \centering
    \includegraphics[width=0.5\linewidth]{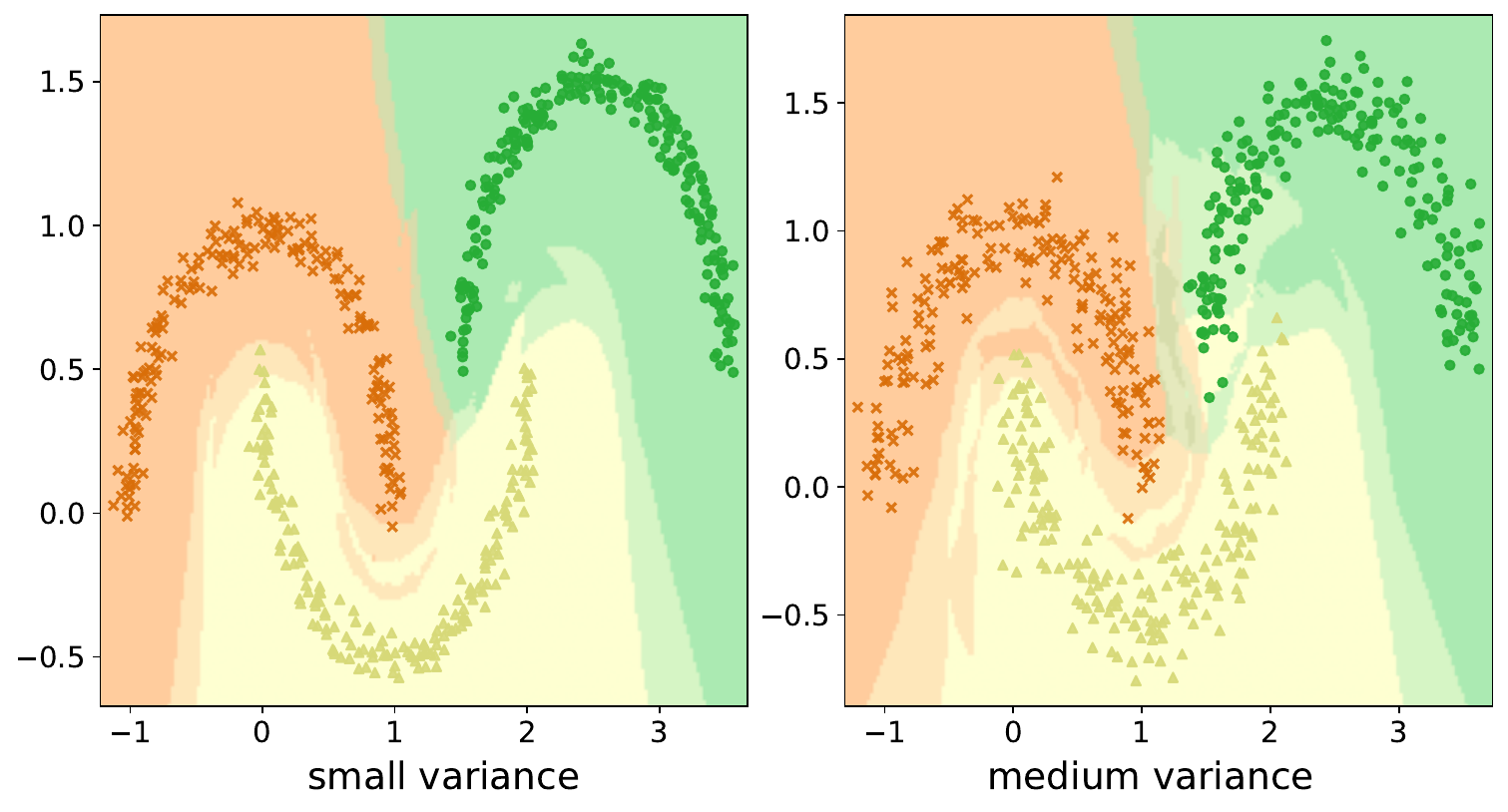}
    \hspace*{-7pt}
    \includegraphics[width=0.5\linewidth]{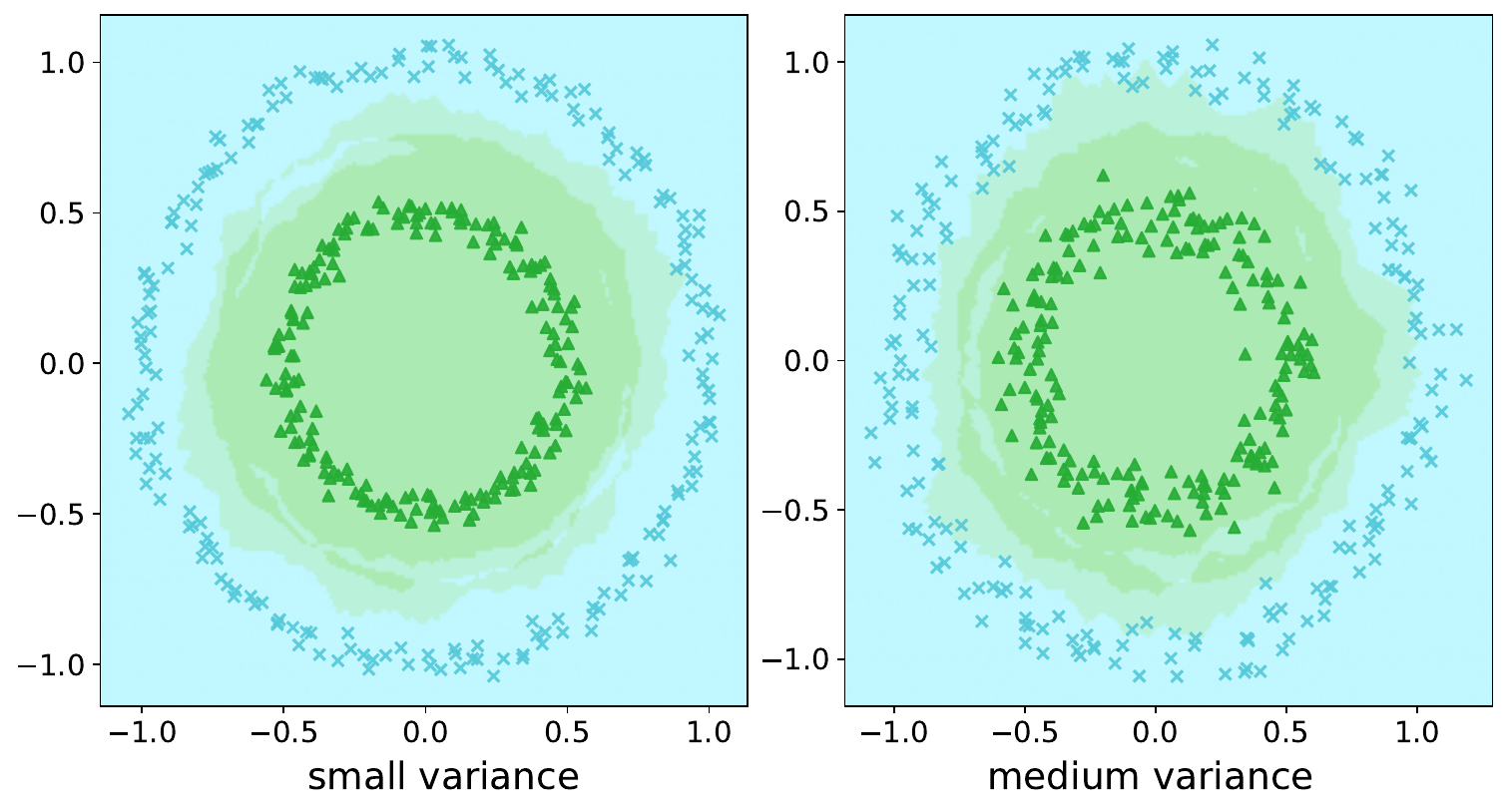}
    \includegraphics[width=\linewidth]{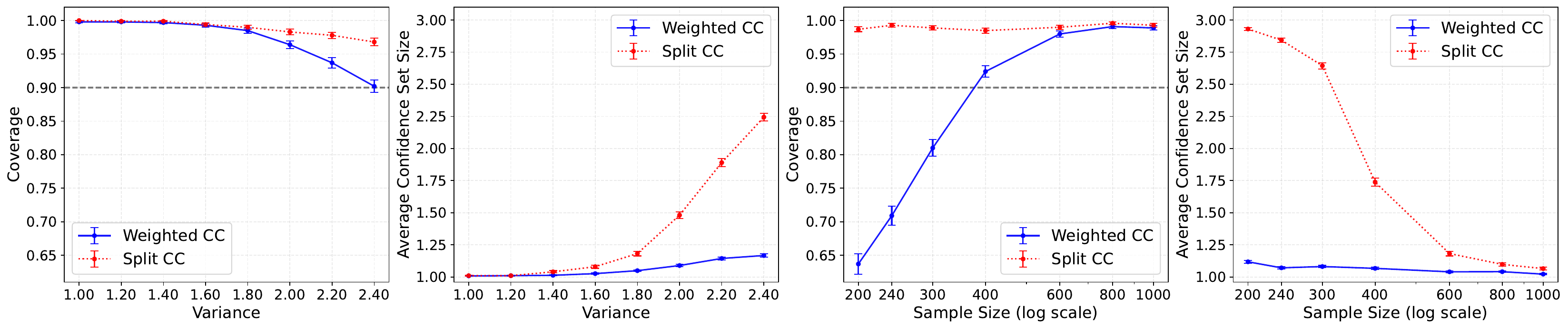}
    \includegraphics[width=\linewidth]{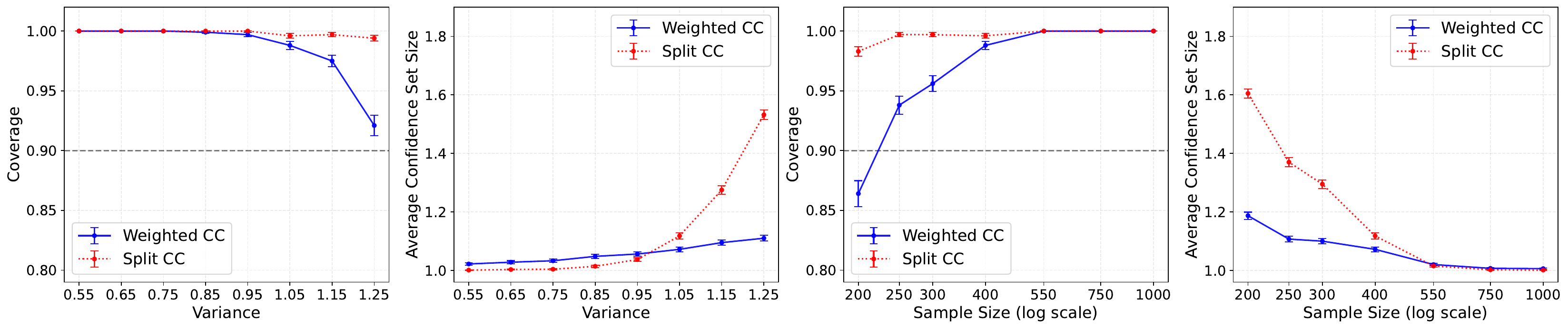}
    \vspace*{-10pt}
    \caption{Confidence set heatmaps and coverage/set size for the nonlinear clustering simulations. Top row: confidence sets produced by Algorithm \ref{alg:estimated_weighted_conformal_clustering} for moons (\(\sigma^2 \in \{1.0, 2.2\}\)) and circles (\(\sigma^2 \in \{0.65, 1.25\}\)). Middle row: coverage and set size for the moon example for varying $\sigma^2$ and for varying \(n\). Bottom row: coverage and set size for the circle example. Algorithm \ref{alg:estimated_weighted_conformal_clustering} and split CC are implemented with spectral clustering that applies stochastic GMM clustering to the top \(K\) singular vectors of the Laplacian of the mutual nearest neighbor graph.}
    \label{fig:2D_non_linear}
\end{figure}

\paragraph{Nonlinear Clusters} We consider two nonlinear clustering examples based on three moons and two concentric circles. In both examples, each component lies on a nonlinear manifold, which is not linearly separable. We apply spectral clustering that applies stochastic GMM clustering to the top \(K\) singular vectors of the Laplacian of the mutual k-nearest neighbor graph (with k = 5\% of \(n_{tr}\)). For these examples, the variance \(\sigma^2\) denotes the spread (thickness) of the manifold around the ideal nonlinear structure. For varying \(\sigma^2\) (with fixed \(n = 600\) for moons, \(n = 400\) for circles) or varying \(n\) (with fixed \(\sigma^2 = 1.8\) for moons, \(\sigma^2 = 1.05\) for circles), we compare the coverage and set size of Algorithm \ref{alg:estimated_weighted_conformal_clustering} and split CC.

The top row of Figure \ref{fig:2D_non_linear} visualizes the confidence sets produced by Algorithm \ref{alg:estimated_weighted_conformal_clustering} as heatmaps. Simulated points are overlaid with the heatmaps that color each point \(x \in \R^2\) by a mixture of the colors corresponding to the labels in the confidence set \(\hat{\cC}(x)\). In both examples, we see that the sets become larger for larger $\sigma^2$ as the gap between the manifolds becomes smaller, making the clustering task more difficult and thus increasing the uncertainty. Meanwhile, the middle and bottom rows of Figure \ref{fig:2D_non_linear} show that both Algorithm \ref{alg:estimated_weighted_conformal_clustering} and split CC have almost perfect coverage with informative set size for sufficiently small $\sigma^2$, while the coverage drops for larger $\sigma^2$ as clustering becomes more challenging. Interestingly, Algorithm \ref{alg:estimated_weighted_conformal_clustering} maintains the informative set size for a wider range of $\sigma^2$ while still achieving valid coverage, but split CC fails to produce informative sets for larger $\sigma^2$. Also, we can see that both methods achieve asymptotically valid coverage with decreasing set size as \(n\) increases, but Algorithm \ref{alg:estimated_weighted_conformal_clustering} achieves this faster than split CC, as also observed in the spectral clustering simulations for the GMM in \(\R^{50}\). 

In summary, Algorithm \ref{alg:estimated_weighted_conformal_clustering} provides valid and more informative confidence sets for nonlinear clustering compared to split CC, confirming that the weighted conformal framework brings benefits to a wide range of clustering problems beyond the standard mixture model setting considered in \cite{hur2026inference}.

\subsection{MNIST Application with Deep Clustering}

\begin{figure}[t]
    \centering\includegraphics[width=\linewidth]{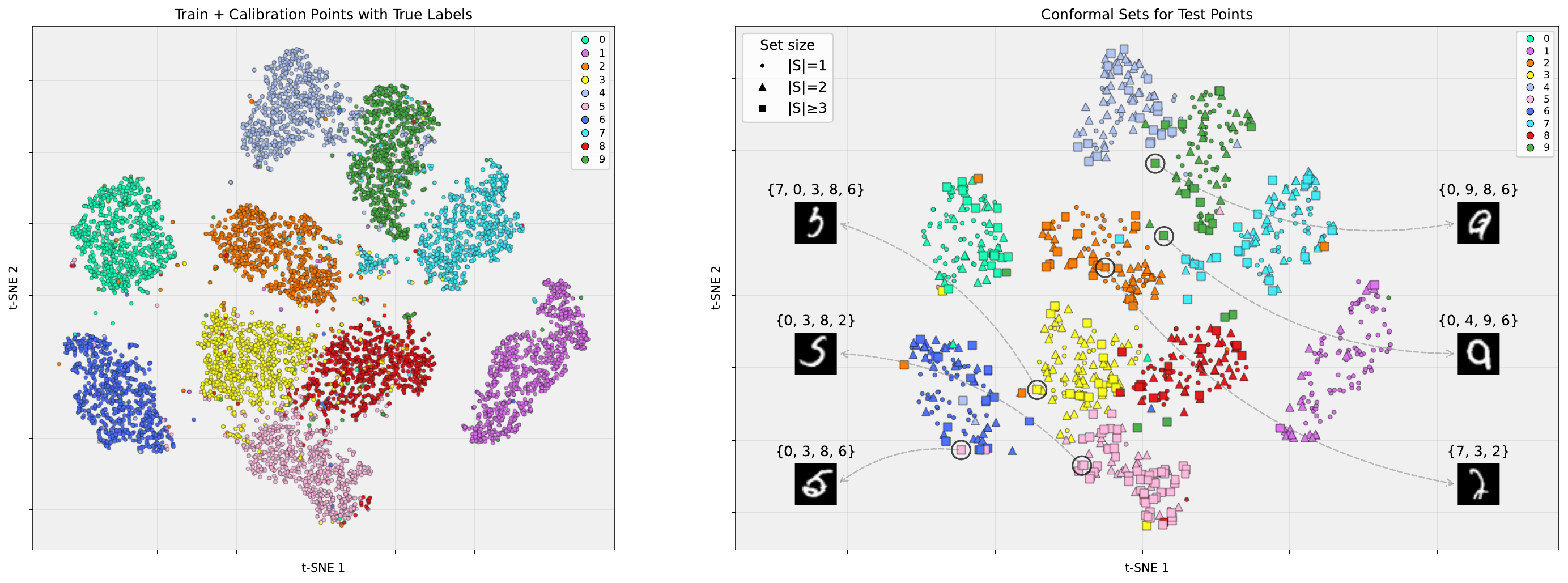}
    \vspace*{-10pt}
    \caption{True labels and confidence sets for the MNIST dataset. Left: true digit labels for the training and calibration sets in the t-SNE space. Right: confidence sets produced by Algorithm \ref{alg:estimated_weighted_conformal_clustering} for the test set, visualized with different markers based on the set size and colored by the true digit labels. Several ambiguous points with large confidence sets are shown with their corresponding images.}
    \label{fig:mnist}
\end{figure}

We apply Algorithm \ref{alg:estimated_weighted_conformal_clustering} with deep clustering to the MNIST dataset, which consists of images of handwritten digits (0--9) in \(\R^{28 \times 28}\). Deep clustering \cite{xie2016unsupervised,caron2018deep,zhou2024comprehensive}, which combines deep representation learning with clustering, has been shown to be effective for clustering high-dimensional data such as images. We first pretrain a convolutional autoencoder to learn a low-dimensional (10D) representation. We focus on 10,000 images not used for pretraining, which we split into training (4,500), calibration (4,500), and test (1,000) sets. For each test image, we produce a confidence set using Algorithm \ref{alg:estimated_weighted_conformal_clustering} with stochastic GMM clustering applied in the 10D representation space with $\alpha = 0.1$. 

The left panel of Figure \ref{fig:mnist} plots the true digit labels for the training and calibration sets in the t-SNE space. The right panel visualizes the confidence sets for the test set, where each point is colored by its true label and marked by its set size (circle, triangle, square for size 1, 2, 3+). Most of them are singleton confidence sets (45.3\%) and tend to be located near the centers of the clusters, while the rest with size 2 (30.4\%) or 3+ (24.3\%) tend to appear near the boundaries between clusters or in more ambiguous regions. Several ambiguous points with large confidence sets are shown with their corresponding images, which are indeed tricky to classify even for humans. 

We can also check the coverage by using the true labels. To do so, we use 10-fold stratified cross-validation, where we calculate the coverage of the true labels of each fold (1,000 points) given the weighted conformal sets produced by the remaining 9 folds. The average coverage across folds is 0.872 with a standard deviation of 0.050, and the average set size across folds is 1.868 with a standard deviation of 0.064. Though slightly below the nominal level of 0.9, this shows that weighted conformal clustering can provide a reasonable coverage.

In summary, this application demonstrates that Algorithm \ref{alg:estimated_weighted_conformal_clustering} can provide informative uncertainty quantification with deep clustering for high-dimensional data, extending its utility beyond standard mixture models and spectral clustering. All the experiments in this section are reproducible using the code available at \url{https://github.com/DataSlingers/Weighted_Conformal_Clustering}.

\section{Discussion}
This paper introduces weighted conformal clustering, an innovative framework that extends the conformal clustering approach of \cite{hur2026inference} by treating synthetic clustering labels as inducing a conditional label-distribution shift from the latent target labels. Using weighted conformal prediction beyond exchangeability, we develop a tractable algorithm with asymptotic coverage under consistent conditional-density estimation, and show empirically that it yields smaller confidence sets in high-dimensional and nonlinear clustering while matching existing methods in smooth mixture settings. This represents a substantive step forward for the emerging conformal clustering framework, with potential implications for conformal prediction in unsupervised learning beyond clustering. A key ingredient is the novel augmented calibration mechanism, which makes weighted conformal prediction computationally tractable while combining aspects of both full and split conformal procedures. This suggests augmented calibration as a broader methodological principle as a whole for more interesting applications in the future.

\section*{Acknowledgments} 
The authors acknowledge funding from NSF DMS-2516872.

\bibliographystyle{plain}
\bibliography{ref}

%%%%%%%%%%%%%%%%%%%%%%%%%%%%%%%%%%%%%%%%%%%%%%%%%%%%%%%%%%%%
\newpage
\appendix

\section{Theory of Weighted Conformal Prediction}
\label{appendix:wcp_theory}

\begin{definition}
    \label{def:quantile}
    For any measure $\mu$ on $\R$, let $\mathsf{Q}_\tau(\mu)$ be the $\tau$-th quantile of $\mu$ defined by
    \begin{equation*}
        \mathsf{Q}_\tau(\mu) = \inf\{t \in \R : \mu((-\infty, t]) \ge \tau\} \quad \forall \tau \in [0, 1].
    \end{equation*}
    For $\{s_1, \ldots, s_n\} \subset \R$, we abbreviate $\mathsf{Q}_\tau\left(\frac{1}{n} \sum_{i = 1}^{n} \delta_{s_i}\right)$ as $\mathsf{Q}_\tau(\{s_i\}_{i = 1}^{n})$.
\end{definition}

\subsection{Full Conformal Prediction with Weighting}
We state the weighted conformal procedure, where $(X_1, Y_1), \ldots, (X_{n + 1}, Y_{n + 1})$ are from an arbitrary distribution $P_{n + 1}$ on $(\cX \times \cY)^{n + 1}$, which may not be a product distribution. The high-level idea is explained in Section 7.5 of \cite{angelopoulos2024theoretical}. We review this here for completeness.

Recall that the full conformal prediction set can be rewritten as
\begin{equation*}
    C_n(x) = \left\{y \in \cY : s((x, y); \cD_n^{(x, y)}) \le \mathsf{Q}_{1 - \alpha}\left(\frac{\delta_{s((x, y); \cD_n^{(x, y)})} + \sum_{i = 1}^{n} \delta_{s((X_i, Y_i); \cD_n^{(x, y)})}}{n + 1}\right)\right\},
\end{equation*}
where $\cD_n^{(x, y)} = \{(X_1, Y_1), \ldots, (X_n, Y_n), (x, y)\}$ and $s \colon \cX \times \cY \times 2^{\cX \times \cY} \to \R$ is a conformity score function.\footnote{The second argument of $s$ is an unordered subset of $\cX \times \cY$, meaning that $s$ ignores the order of its arguments.} The full conformal prediction set above no longer satisfies the coverage guarantee because the exchangeability assumption is violated. Particularly, the score $s((X_{n + 1}, Y_{n + 1}); \cD_n^{(X_{n + 1}, Y_{n + 1})})$ may tend to be larger than the scores of the training data because of distribution shift, and hence the prediction set will tend to under-cover. 

Let \(Z_1,\ldots,Z_n,Z_{n+1}\in\mathcal Z := \cX \times \cY\) have joint density \(f\) on \(\mathcal Z^{n+1}\), without assuming exchangeability, and let \(\hat P_{n+1}=\{Z_1,\ldots,Z_{n+1}\}\) denote the unordered empirical measure. Conditional on \(\hat P_{n+1}\), the identity of the test point is determined by the relative likelihoods of the different orderings of these \(n+1\) observations under \(f\). Now, fix a candidate label \(y\), and form the augmented sample by setting \(Z_i^y=Z_i\) for \(i\in[n]\) and \(Z_{n+1}^y=(X_{n+1},y)\). The following theorem states the general weighted conformal procedure and its marginal validity.

\begin{theorem}
    \label{thm:weighted_conformal_permutation}
    Suppose $(X_1, Y_1), \ldots, (X_{n + 1}, Y_{n + 1})$ are drawn from some distribution $P_{n + 1}$ on $(\cX \times \cY)^{n + 1}$ whose probability density function is $f \colon (\cX \times \cY)^{n + 1} \to \R_+$. For $(x, y) \in \cX \times \cY$, let $Z_i^{(x, y)} := (X_i, Y_i)$ for $i \in [n]$ and $Z_{n + 1}^{(x, y)} := (x, y)$, and for any permutation $\sigma$ of $[n + 1]$, define
    \begin{equation*}
        w_\sigma^{(x, y)} := \frac{f(Z_{\sigma(1)}^{(x, y)}, \ldots, Z_{\sigma(n + 1)}^{(x, y)})}{\sum_{\sigma' \in \cS_{n + 1}} f(Z_{\sigma'(1)}^{(x, y)}, \ldots, Z_{\sigma'(n + 1)}^{(x, y)})}.
    \end{equation*} 
    Let $s \colon \cX \times \cY \times 2^{\cX \times \cY} \to \R$ be a fixed function, where the second argument of $s$ is an unordered subset of $\cX \times \cY$. Define
    \begin{equation*}
        C_n(x) = \left\{y \in \cY : s((x, y); \cD_n^{(x, y)}) \le \mathsf{Q}_{1 - \alpha}\left(\sum_{\sigma \in \cS_{n + 1}} w_\sigma^{(x, y)} \delta_{s(Z_{\sigma(n + 1)}^{(x, y)}; \cD_n^{(x, y)})}\right)\right\},
    \end{equation*}
    where $\cD_n^{(x, y)} = \{(X_1, Y_1), \ldots, (X_n, Y_n), (x, y)\}$. Then, we have
    \begin{equation*}
        \P(Y_{n + 1} \in C_n(X_{n + 1})) \ge 1 - \alpha.
    \end{equation*}
\end{theorem}
\begin{proof}
    Let $\cD_{n + 1} = \cD_n^{(X_{n + 1}, Y_{n + 1})}$. Let $S_i = s((X_i, Y_i); \cD_{n + 1})$ for $i = 1, \ldots, n + 1$. Then, by definition, we have
    \begin{equation*}
        Y_{n + 1} \in C_n(X_{n + 1}) \Leftrightarrow S_{n + 1} \le \mathsf{Q}_{1 - \alpha}\left(\sum_{\sigma \in \cS_{n + 1}} w_{\sigma} \delta_{S_{\sigma(n + 1)}}\right) =: q,
    \end{equation*}
    where $w_\sigma = w_\sigma^{(X_{n + 1}, Y_{n + 1})}$. By Lemma \ref{lem:conditional_distribution_permutation}, we have
    \begin{equation*}
        S_{n + 1} \, \Big| \, \frac{1}{n + 1} \sum_{i = 1}^{n + 1} \delta_{(X_i, Y_i)} \sim \sum_{\sigma \in \cS_{n + 1}} w_\sigma \delta_{S_{\sigma(n + 1)}}.
    \end{equation*}
    Hence, 
    \begin{equation*}
        \P\left(S_{n + 1} \le q \, | \, \frac{1}{n + 1} \sum_{i = 1}^{n + 1} \delta_{(X_i, Y_i)}\right) = \sum_{\sigma \in \cS_{n + 1}} w_\sigma 1\{S_{\sigma(n + 1)} \le q\} \ge 1 - \alpha,
    \end{equation*}
    where the last inequality is due to Lemma \ref{lem:weighted_quantile_rank}. By the law of total probability, we have
    \begin{equation*}
        \P(Y_{n + 1} \in C_n(X_{n + 1})) \ge 1 - \alpha.
    \end{equation*}
\end{proof}

\subsection{Technical Lemmas}
\begin{lemma}
    \label{lem:weighted_quantile_rank}
    For any $z_1, \ldots, z_n \in \R$, weights $w_1, \ldots, w_n \ge 0$ with $\sum_{i = 1}^{n} w_i = 1$, and $\tau \in [0, 1]$, we have
    \begin{equation*}
        \sum_{i = 1}^{n} w_i 1\left\{z_i \le \mathsf{Q}_\tau\left(\sum_{j = 1}^{n} w_j \delta_{z_j}\right)\right\} \ge \tau.
    \end{equation*}
    Suppose we swap $z_n$ with any $z \in \R \cup \{\infty\}$ such that $z \ge \max_{i \in [n]} z_i$. Then, we have
    \begin{equation*}
        z_n \le \mathsf{Q}_\tau\left(\sum_{i = 1}^{n} w_i \delta_{z_i}\right) \quad \Leftrightarrow \quad z_n \le \mathsf{Q}_\tau\left(\sum_{i = 1}^{n - 1} w_i \delta_{z_i} + w_n \delta_z\right). 
    \end{equation*}
\end{lemma}
\begin{proof}
    Let $F$ be the CDF of the measure $\sum_{i = 1}^{n} w_i \delta_{z_i}$, i.e., $F(t) = \sum_{i = 1}^{n} w_i 1\{z_i \le t\}$ for $t \in \R$. By definition, we have $\mathsf{Q}_\tau\left(\sum_{i = 1}^{n} w_i \delta_{z_i}\right) = \inf\{t \in \R : F(t) \ge \tau\}$. By the right-continuity of $F$, we can deduce that $F\left(\mathsf{Q}_\tau\left(\sum_{i = 1}^{n} w_i \delta_{z_i}\right)\right) \ge \tau$, which is exactly the desired result. Next, for $z \ge \max_{i \in [n]} z_i$, we have 
    \begin{equation*}
        q := \mathsf{Q}_\tau\left(\sum_{i = 1}^{n} w_i \delta_{z_i}\right) \le \mathsf{Q}_\tau\left(\sum_{i = 1}^{n - 1} w_i \delta_{z_i} + w_n \delta_z\right) =: q^z.
    \end{equation*}
    Hence, $z_n \le q$ implies $z_n \le q^z$. Conversely, suppose $z_n \le q^z$. Let $F^z$ be the CDF of the measure $\sum_{i = 1}^{n - 1} w_i \delta_{z_i} + w_n \delta_z$. Since $z_n \le q^z$, we must have $F^z(t) < \tau$ for any $t < z_n$; also, we have $F(t) = F^z(t)$ for any $t < z_n$ because $z_n \le z$ implies
    \begin{align*}
        F(t) & = \sum_{i = 1}^{n - 1} w_i 1\{z_i \le t\} + w_n 1\{z_n \le t\} = \sum_{i = 1}^{n - 1} w_i 1\{z_i \le t\}, \\
        F^z(t) & = \sum_{i = 1}^{n - 1} w_i 1\{z_i \le t\} + w_n 1\{z \le t\} = \sum_{i = 1}^{n - 1} w_i 1\{z_i \le t\}.
    \end{align*}
    Hence, for any $t < z_n$, we have $F(t) = F^z(t) < \tau$, which implies $z_n \le q$.
\end{proof}

\begin{lemma}
    \label{lem:conditional_distribution_permutation}
    Suppose $\cZ$ is a finite set. Let $P_n$ be a distribution on $\cZ^n$ with the probability mass function $f \colon \cZ^n \to \R_+$. For $(Z_1, \ldots, Z_n) \sim P_n$, we have 
    \begin{equation*}
        Z_i \, | \, \{Z_1, \ldots, Z_n\} = \{z_1, \ldots, z_n\} \sim \frac{1}{\sum_{\sigma \in \cS_n} f(z_{\sigma(1)}, \ldots, z_{\sigma(n)})} \sum_{\sigma \in \cS_n} f(z_{\sigma(1)}, \ldots, z_{\sigma(n)}) \delta_{z_{\sigma(i)}}
    \end{equation*}
    for any $i \in [n]$ and $z_1, \ldots, z_n \in \cZ$. If we have a function $s \colon \cZ \times 2^\cZ \to S$, where $S$ is some set, then we also have
    \begin{equation*}
        \begin{split}
            & s(Z_i; \{Z_1, \ldots, Z_n\}) \, | \, \{Z_1, \ldots, Z_n\} = \{z_1, \ldots, z_n\} \\
            & \qquad \qquad \sim \frac{1}{\sum_{\sigma \in \cS_n} f(z_{\sigma(1)}, \ldots, z_{\sigma(n)})} \sum_{\sigma \in \cS_n} f(z_{\sigma(1)}, \ldots, z_{\sigma(n)}) \delta_{s(z_{\sigma(i)}; \{z_1, \ldots, z_n\})}.
        \end{split}
    \end{equation*}
    If $P_n = P \otimes \cdots \otimes P \otimes Q$ for some distributions $P, Q$ on $\cZ$ whose probability mass functions are $p, q \colon \cZ \to \R_+$, respectively, then
     we have 
    \begin{equation*}
        Z_n \, | \, \{Z_1, \ldots, Z_n\} = \{z_1, \ldots, z_n\} \sim \frac{1}{\sum_{i = 1}^{n} \frac{q\{z_i\}}{p\{z_i\}}} \sum_{i = 1}^{n} \frac{q\{z_i\}}{p\{z_i\}} \delta_{z_i}
    \end{equation*}
    and 
    \begin{equation*}
        s(Z_i; \{Z_1, \ldots, Z_n\}) \, | \, \{Z_1, \ldots, Z_n\} = \{z_1, \ldots, z_n\} \sim \frac{1}{\sum_{i = 1}^{n} \frac{q\{z_i\}}{p\{z_i\}}} \sum_{i = 1}^{n} \frac{q\{z_i\}}{p\{z_i\}} \delta_{s(z_{\sigma(i)}; \{z_1, \ldots, z_n\})}.
    \end{equation*}
\end{lemma}
\begin{proof}
    Without loss of generality, we can assume $i = n$. It suffices to show that for any $A \subset \cZ$, we have
    \begin{equation*}
        \P(Z_n \in A \, | \, \{Z_1, \ldots, Z_n\} = \{z_1, \ldots, z_n\}) = \frac{\sum_{\sigma \in \cS_n} f(z_{\sigma(1)}, \ldots, z_{\sigma(n)}) 1\{z_{\sigma(n)} \in A\}}{\sum_{\sigma \in \cS_n} f(z_{\sigma(1)}, \ldots, z_{\sigma(n)})}.
    \end{equation*}
    Clearly, 
    \begin{equation*}
        \begin{split}
            \P(\{Z_1, \ldots, Z_n\} = \{z_1, \ldots, z_n\})
            & = \sum_{\sigma \in \cS_n} \P((Z_1, \ldots, Z_n) = (z_{\sigma(1)}, \ldots, z_{\sigma(n)})) \\
            & = \sum_{\sigma \in \cS_n} f(z_{\sigma(1)}, \ldots, z_{\sigma(n)}).
        \end{split}
    \end{equation*}
    Similarly, 
    \begin{equation*}
        \begin{split}
            & \P(Z_n \in A, \{Z_1, \ldots, Z_n\} = \{z_1, \ldots, z_n\}) \\
            & \quad = \sum_{\sigma \in \cS_n} \P((Z_1, \ldots, Z_n) = (z_{\sigma(1)}, \ldots, z_{\sigma(n)}), z_{\sigma(n)} \in A) \\
            & \quad = \sum_{\sigma \in \cS_n} f(z_{\sigma(1)}, \ldots, z_{\sigma(n)}) 1\{z_{\sigma(n)} \in A\}.
        \end{split}
    \end{equation*}
    Hence, we prove the desired claim. Meanwhile, for a function $s \colon \cZ \times 2^\cZ \to S$, where $S$ is some set, we have for any $B \subset S$,
    \begin{equation*}        
        \begin{split}
            & \P(s(Z_n; \{Z_1, \ldots, Z_n\}) \in B, \{Z_1, \ldots, Z_n\} = \{z_1, \ldots, z_n\}) \\
            & \quad = \sum_{\sigma \in \cS_n} \P((Z_1, \ldots, Z_n) = (z_{\sigma(1)}, \ldots, z_{\sigma(n)}), s(z_{\sigma(n)}; \{z_1, \ldots, z_n\}) \in B) \\
            & \quad = \sum_{\sigma \in \cS_n} f(z_{\sigma(1)}, \ldots, z_{\sigma(n)}) 1\{s(z_{\sigma(n)}; \{z_1, \ldots, z_n\}) \in B\}.
        \end{split}
    \end{equation*}
    Hence, we prove the second claim. Lastly, when $P_n = P \otimes \cdots \otimes P \otimes Q$, we have
    \begin{equation*}
        f(z_1, \ldots, z_n) = \prod_{i = 1}^{n - 1} p\{z_i\} \cdot q\{z_n\} = \left(\prod_{i = 1}^{n} p\{z_i\}\right) \cdot \frac{q\{z_n\}}{p\{z_n\}}.
    \end{equation*}
    Hence, for any $\sigma \in \cS_n$, we have
    \begin{equation*}
        f(z_{\sigma(1)}, \ldots, z_{\sigma(n)}) = \left(\prod_{i = 1}^{n} p\{z_i\}\right) \cdot \frac{q\{z_{\sigma(n)}\}}{p\{z_{\sigma(n)}\}}.
    \end{equation*}
    Therefore, 
    \begin{equation*}
        \begin{split}
            \sum_{\sigma \in \cS_n} f(z_{\sigma(1)}, \ldots, z_{\sigma(n)}) 1\{z_{\sigma(n)} \in A\} 
            & = \left(\prod_{i = 1}^{n} p\{z_i\}\right) \cdot \sum_{\sigma \in \cS_n} \frac{q\{z_{\sigma(n)}\}}{p\{z_{\sigma(n)}\}} 1\{z_{\sigma(n)} \in A\} \\
            & = \left(\prod_{i = 1}^{n} p\{z_i\}\right) \cdot \sum_{i = 1}^{n} \sum_{\sigma \in \cS_n, \sigma(n) = i} \frac{q\{z_i\}}{p\{z_i\}} 1\{z_i \in A\} \\
            & = (n - 1)! \left(\prod_{i = 1}^{n} p\{z_i\}\right) \cdot \sum_{i = 1}^{n} \frac{q\{z_i\}}{p\{z_i\}} 1\{z_i \in A\}.
        \end{split}
    \end{equation*}
    Similarly, 
    \begin{equation*}
        \sum_{\sigma \in \cS_n} f(z_{\sigma(1)}, \ldots, z_{\sigma(n)}) = (n - 1)! \left(\prod_{i = 1}^{n} p\{z_i\}\right) \cdot \sum_{i = 1}^{n} \frac{q\{z_i\}}{p\{z_i\}}.
    \end{equation*}
    Hence, 
    \begin{equation*}
        \P(Z_n \in A \, | \, \{Z_1, \ldots, Z_n\} = \{z_1, \ldots, z_n\}) = \frac{\sum_{i = 1}^{n} \frac{q\{z_i\}}{p\{z_i\}} 1\{z_i \in A\}}{\sum_{i = 1}^{n} \frac{q\{z_i\}}{p\{z_i\}}},
    \end{equation*}
    which proves
    \begin{equation*}
        Z_n \, | \, \{Z_1, \ldots, Z_n\} = \{z_1, \ldots, z_n\} \sim \frac{1}{\sum_{i = 1}^{n} \frac{q\{z_i\}}{p\{z_i\}}} \sum_{i = 1}^{n} \frac{q\{z_i\}}{p\{z_i\}} \delta_{z_i}.
    \end{equation*}
\end{proof}

\begin{remark}
    Lemma \ref{lem:conditional_distribution_permutation} can be extended to the case where $\cZ$ is an infinite set with $f$ being a probability density function. In this case, instead of conditioning on the event $\{Z_1, \ldots, Z_n\} = \{z_1, \ldots, z_n\}$, we should think of the regular conditional distribution of $Z_i$ given a random measure $\frac{1}{n} \sum_{j = 1}^{n} \delta_{Z_j}$ on $\cZ$. Also, the last part of Lemma \ref{lem:conditional_distribution_permutation} still holds with $\frac{\mathrm{d} Q}{\mathrm{d} P}(z_i)$ in the place of $\frac{q\{z_i\}}{p\{z_i\}}$. This is formally stated as Proposition 7.6 of \cite{angelopoulos2024theoretical}.
\end{remark}

\section{Proofs}
\label{sec:proofs}
\subsection{Proof of Theorem \ref{thm:exact_weights}}

\begin{proof}
    By the definition of the alignment $\hat{\sigma}_o^\ast$, we have 
    \begin{equation*}
        \P\left(\hat{\sigma}_o^\ast(Y_{n + 1}^\ast) \in \hat{\cC}(X_{n + 1}) \mid \hat{\cC}\right) \ge \P\left(Y_{n + 1}^\ast \in \hat{\cC}(X_{n + 1}) \mid \hat{\cC}\right).
    \end{equation*}
    By taking the expectation, we have $\P(\hat{\sigma}_o^\ast(Y_{n + 1}^\ast) \in \hat{\cC}(X_{n + 1})) \ge \P(Y_{n + 1}^\ast \in \hat{\cC}(X_{n + 1}))$. Hence, it suffices to show $\P(Y_{n + 1}^\ast \in \hat{\cC}(X_{n + 1})) \ge 1 - \alpha$. 
    
    For notational convenience, we prove this for a pretrained version of Algorithm~\ref{alg:oracle_weighted_conformal_clustering}, where we assume that $\cI_{ca} = \{1, \ldots, n\}$, namely, take the whole input as the calibration, while treating the training split and the soft classifier trained on it as fixed independent quantities. Accordingly, from now on, we will condition on the training split. 
    
    Now, let
    \[
        Z_i := (X_i,Y_i), \; \ i=1,\ldots,n,
        \qquad\text{and}\qquad
        Z_{n+1} := (X_{n+1}, Y_{n + 1}^\ast).
    \]
    We will derive the oracle weights appearing in Algorithm~\ref{alg:oracle_weighted_conformal_clustering} from the general weighted conformal prediction theory presented in Appendix~\ref{appendix:wcp_theory}. Note that \(X_1,\ldots,X_{n+1}\) are i.i.d.\ from \(P_X^\ast\), and conditional on the calibration covariates, the synthetic calibration labels are independent with
    \[
        Y_j \mid X_1,\ldots,X_n \sim \mathrm{Cat}\bigl(\mathcal A_{\cD_{ca}}(X_j)\bigr),
        \qquad j=1,\ldots,n,
    \]
    while the test label $Y_{n + 1}^\ast$ is drawn from the true conditional law \(P_{Y\mid X}^\ast(\cdot \mid X_{n+1})\). Then, the joint density of the augmented sample $(Z_1,\ldots,Z_{n+1})$ at $(z_1, \ldots, z_{n + 1}) = ((x_1, y_1), \ldots, (x_{n + 1}, y_{n + 1}))$
    is
    \begin{equation}
        \label{eq:joint_density_original}
        f(z_1,\ldots,z_{n+1})
        =
        \left(\prod_{k=1}^{n+1} p_X^\ast(x_k)\right)
        \left(\prod_{j=1}^n \left[\mathcal A_{\{x_1, \ldots, x_n\}}(x_j)\right]_{y_j}\right)
        p_{Y\mid X}^\ast(y_{n + 1} \mid x_{n+1}).
    \end{equation}
    To compute the conformal weights, consider the permutations such that the \(i\)-th observation occupies the test position, namely, \(\sigma \in \mathcal S_{n+1}\) such that $\sigma(n+1)=i$. Under this permutation, the point \((x_i,y_i)\) is treated as the test point, while the remaining \(n\) covariates are $\{x_{\sigma(1)},\ldots,x_{\sigma(n)}\}$. Thus, substituting the permuted indices into the joint density gives
    \begin{equation*}
        f(z_{\sigma(1)},\ldots,z_{\sigma(n+1)})
        =
        \left(\prod_{k=1}^{n+1} p_X^\ast(x_{\sigma(k)})\right)
        \left(\prod_{j=1}^n \left[\mathcal A_{\{x_{\sigma(1)},\ldots,x_{\sigma(n)}\}}(x_{\sigma(j)})\right]_{y_{\sigma(j)}}\right)
        p_{Y\mid X}^\ast(y_i \mid x_i).
    \end{equation*}
    Let \(\{x_{\sigma(1)},\ldots,x_{\sigma(n)}\} = \{x_1, \ldots, x_{n + 1}\} \setminus \{x_i\} =: \cD_{n+1}^{-i}\). Then,
    \[
        \prod_{j=1}^n \left[\mathcal A_{\{x_{\sigma(1)},\ldots,x_{\sigma(n)}\}}(x_{\sigma(j)})\right]_{y_{\sigma(j)}}
        =
        \prod_{j \in [n + 1] \setminus \{i\}} \left[\mathcal A_{\cD_{n+1}^{-i}}(x_j)\right]_{y_j}.
    \]
    Also, note that $\prod_{k=1}^{n+1} p_X^\ast(x_{\sigma(k)})=\prod_{k=1}^{n+1} p_X^\ast(x_k)$. Thus, we have
    \begin{equation}
        \label{eq:joint_density_permuted_simplified}
        f(z_{\sigma(1)},\ldots,z_{\sigma(n+1)})
        =
        \left(\prod_{k=1}^{n+1} p_X^\ast(x_k)\right)
        p_{Y\mid X}^\ast(y_i \mid x_i)
        \prod_{j \in [n + 1] \setminus \{i\}} \left[\mathcal A_{\cD_{n+1}^{-i}}(x_j)\right]_{y_j}.
    \end{equation}
    Observe that the right-hand side of \eqref{eq:joint_density_permuted_simplified} depends only on the index \(i\), and not on the particular choice of permutation \(\sigma\) satisfying \(\sigma(n+1)=i\). Since there are exactly \(n!\) such permutations, this yields the unnormalized weights as
    \begin{equation}
    \label{eq:unnormalized_weights}
    \begin{split}
        \tilde w_i 
        & := \sum_{\sigma \in \mathcal{S}_{n+1}:\sigma(n+1)=i} f(z_{\sigma(1)},\ldots,z_{\sigma(n+1)}) \\
        & =
        n!\left(\prod_{k=1}^{n+1} p_X^\ast(x_k)\right)
        p_{Y\mid X}^\ast(y_i \mid x_i)
        \prod_{j \in [n + 1] \setminus \{i\}} \left[\mathcal A_{\cD_{n+1}^{-i}}(x_j)\right]_{y_j},
    \end{split}
    \end{equation}
    for $i=1,\ldots,n+1$. Now, replacing $(x_i, y_i)$'s with $(X_i, Y_i)$'s and $(x_{n + 1}, y_{n + 1})$ with $(X_{n + 1}, y) = (x, y)$, we obtain the unnormalized weights in \eqref{eq:weight_calibration}-\eqref{eq:weight_test}. Now, we can apply Theorem~\ref{thm:weighted_conformal_permutation} to show that the probability of $Y_{n + 1}^\ast \in \hat{\cC}(X_{n + 1})$ conditional on the training split is at least $1 - \alpha$. By the tower property, we complete the proof.
\end{proof}

\begin{remark}
    \label{rmk:why_no_alignment}
    As we can see from the proof, the key is to apply Theorem~\ref{thm:weighted_conformal_permutation} after calculating the weights. This is possible because of the scores defined in Algorithm~\ref{alg:oracle_weighted_conformal_clustering}. If we included the alignment step as in Algorithm 1 of \cite{hur2026inference}, we would end up with alignment-dependent scores, which we cannot apply Theorem~\ref{thm:weighted_conformal_permutation} to. 
\end{remark}

\subsection{Proof of Theorem \ref{thm:coverage_estimated_weights}}

\begin{proof}
    As in the proof of Theorem \ref{thm:exact_weights}, we have $\P(\hat{\sigma}_o^\ast(Y_{n + 1}^\ast) \in \hat{\cC}(X_{n + 1})) \ge \P(Y_{n + 1}^\ast \in \hat{\cC}(X_{n + 1}))$. Hence, we derive a lower bound on $\P(Y_{n + 1}^\ast \in \hat{\cC}(X_{n + 1}))$. Also, as in the proof of Theorem \ref{thm:exact_weights}, we assume that $\cI_{ca} = \{1, \ldots, n\}$, treat the training split-related objects as fixed independent quantities, and consider the conditional probability given the training split.
    
    We first consider an auxiliary data-generating process in which the test label is drawn from the estimated conditional law $\widehat{P}_{Y \mid X}$ having the density $\hat{p}_{Y \mid X}$, namely,
    \[
        (X_{n+1}, \tilde{Y}_{n+1})
        \sim
        P_X^\ast \times \widehat P_{Y\mid X},
    \]
    independently of the calibration data. The calibration labels are still generated according to the augmented clustering mechanism. Thus, conditional on $\cD_{aug}$, we have
    \[
        Y_j \mid X_1,\ldots,X_n,X_{n+1}
        \sim
        \mathrm{Cat}\bigl(\mathcal A_{\cD_{aug}}(X_j)\bigr),
        \qquad j=1,\ldots,n.
    \]
    Under this auxiliary setup, let us derive the joint density of the augmented sample $(Z_1,\ldots,Z_{n+1}) := ((X_1, Y_1), \ldots, (X_{n + 1}, \hat{\sigma}^{-1}(\tilde{Y}_{n + 1}))$. Here, notice that $\hat{\sigma}$ in Step 5 of Algorithm~\ref{alg:estimated_weighted_conformal_clustering} is symmetric with respect to the ordering of $X_1, \ldots, X_{n + 1}$. In other words, $\hat{\sigma} = \rho(X_1, \ldots, X_{n + 1})$ for a deterministic symmetric function $\rho$. The density of $(Z_1,\ldots,Z_{n+1})$ at $(z_1, \ldots, z_{n + 1}) = ((x_1, y_1), \ldots, (x_{n + 1}, y_{n + 1}))$ is 
    \begin{equation*}
        f(z_1,\ldots,z_{n+1})
        :=
        \left(\prod_{k=1}^{n+1}p_X^\ast(x_k)\right)
        \left(\prod_{j=1}^n [\mathcal A_{\cD}(x_j)]_{y_j}\right)
        \hat{p}_{Y\mid X}(\hat{\sigma}(y_{n + 1}) \mid x_{n+1}),
    \end{equation*}
    where $\cD := \{x_1, \ldots, x_{n + 1}\}$.
    
    To compute the conformal weights, consider the permutations such that the \(i\)-th observation occupies the test position i.e. \(\sigma \in \mathcal S_{n+1}\) such that $\sigma(n+1)=i$.
    Under this permutation, the point \((x_i,y_i)\) is treated as the test point, while the remaining \(n\) covariates form the calibration set. Since \(\{x_{\sigma(1)},\ldots,x_{\sigma(n+1)}\} = \cD\) and $\prod_{k=1}^{n+1} p_X^\ast(x_{\sigma(k)})
    =\prod_{k=1}^{n+1} p_X^\ast(x_k)$, we have for every \(\sigma\) satisfying \(\sigma(n+1)=i\),
    \begin{equation}
        \label{eq:joint_density_augmented_permuted_simplified}
        f(z_{\sigma(1)},\ldots,z_{\sigma(n+1)})
        =
        \left(\prod_{k=1}^{n+1} p_X^\ast(x_k)\right)
        \left(\prod_{j \in [n + 1] \setminus \{i\}} [\mathcal A_{\cD}(x_j)]_{y_j}\right)
        \hat{p}_{Y\mid X}(\hat{\sigma}(y_i)\mid x_i).
    \end{equation}
    Now observe that the right-hand side of \eqref{eq:joint_density_augmented_permuted_simplified} depends only on the index \(i\), and not on the particular permutation \(\sigma\). Since there are exactly \(n!\) permutations in \(\mathcal S_{n+1}\) satisfying \(\sigma(n+1)=i\), we have
    \begin{equation*}
        \begin{split}
            \tilde w_i 
            & := \sum_{\sigma \in \mathcal{S}_{n+1}:\sigma(n+1)=i} f(z_{\sigma(1)},\ldots,z_{\sigma(n+1)}) \\
            & =
            n!\left(\prod_{k=1}^{n+1} p_X^\ast(x_k)\right)
            \hat{p}_{Y\mid X}(\hat{\sigma}(y_i) \mid x_i)
            \prod_{j \in [n + 1] \setminus \{i\}} \left[\mathcal A_{\cD}(x_j)\right]_{y_j} \\
            & = n!\left(\prod_{k=1}^{n+1} p_X^\ast(x_k)\right) \prod_{j= 1}^{n + 1} \left[\mathcal A_{\cD}(x_j)\right]_{y_j} \frac{\hat{p}_{Y\mid X}(\hat{\sigma}(y_i) \mid x_i)}{ \left[\mathcal A_{\cD}(x_i)\right]_{y_i}}
        \end{split}
    \end{equation*}
    Now, replacing $(x_i, y_i)$'s with $(X_i, Y_i)$'s and $(x_{n + 1}, y_{n + 1})$ with $(X_{n + 1}, y)$, which replaces $\cD$ with $\cD_{aug}$, we have
    \begin{equation}
        \label{eq:estimated_augmented_weights_temp}
        \hat{w}_i
        \propto
        \frac{\hat p_{Y\mid X}(\hat{\sigma}(Y_i)\mid X_i)}
        {\left[\mathcal A_{\cD_{aug}}(X_i)\right]_{Y_i}}
        \quad \text{for} ~ i \in \cI_{ca},
        \qquad
        \tilde{w}(y)
        \propto
        \frac{\hat p_{Y\mid X}(\hat{\sigma}(y) \mid x)}
        {\left[\mathcal A_{\cD_{aug}}(x)\right]_y}.
    \end{equation}
    Note that the weight $\tilde{w}(y)$ for the candidate label differs from $\hat{w}(y)$ in \eqref{eq:estimated_augmented_weights}; in fact, we have $\tilde{w}(y) = \hat{w}(\hat{\sigma}(y))$. Therefore, once we normalize the weights in \eqref{eq:estimated_augmented_weights_temp}, we have \(\{w_i(\hat{\sigma}(y))\}_{i \in \cI_{ca}} \cup \{w(\hat{\sigma}(y))\}\). Now, we can apply Theorem \ref{thm:weighted_conformal_permutation} to $(Z_1,\ldots,Z_{n+1}) := ((X_1, Y_1), \ldots, (X_{n + 1}, \hat{\sigma}^{-1}(\tilde{Y}_{n + 1}))$ with the weights \(\{w_i(\hat{\sigma}(y))\}_{i \in \cI_{ca}} \cup \{w(\hat{\sigma}(y))\}\) and the scores 
    \begin{equation*}
        s((X_1, \hat{\sigma}(Y_1)); \hat{\pi}), \ldots, s((X_n, \hat{\sigma}(Y_n)); \hat{\pi}), s((X_{n + 1}, \hat{\sigma}(\hat{\sigma}^{-1}(\tilde{Y}_{n + 1}))); \hat{\pi}),
    \end{equation*}
    which yield
    \begin{equation}
        \label{eq:auxiliary_coverage_estimated_temp}
        \P\left(\hat{\sigma}^{-1}(\tilde{Y}_{n+1}) \in \tilde{\cC}(X_{n+1})\right) \ge 1-\alpha,
    \end{equation}
    where the set $\tilde{\cC}$ is defined based on Theorem \ref{thm:weighted_conformal_permutation} as follows:
    \begin{equation*}
        \tilde{\cC}(x) = \left\{y \in [K] : s((x, \hat{\sigma}(y)); \hat{\pi}) \le \mathsf{Q}_{1 - \alpha}\left({\textstyle\sum_{i \in \cI_{ca}}} w_i(\hat{\sigma}(y)) \, \delta_{S_i} + w(\hat{\sigma}(y)) \delta_{\infty}\right)\right\}.
    \end{equation*}
    Note that the set $\hat{\cC}(x)$ in Algorithm \ref{alg:estimated_weighted_conformal_clustering} satisfies $\hat{\cC}(x) = \{y \in [K] : \hat{\sigma}^{-1}(y) \in \tilde{\cC}(x)\}$. Hence, \eqref{eq:auxiliary_coverage_estimated_temp} is equivalent to
    \begin{equation}
        \label{eq:auxiliary_coverage_estimated}
        \P\left(\tilde{Y}_{n+1} \in \hat{\cC}(X_{n+1})\right) \ge 1-\alpha.
    \end{equation}
    Now, by the definition of total variation distance,
    \begin{equation}
        \label{eq:tv_transfer_estimated}
        \left|
        \P\left(\tilde{Y}_{n+1}\in \hat{\cC}(X_{n+1})\right) - \P\left(
        Y_{n+1}^\ast\in \hat{\cC}(X_{n+1})\right)\right|
        \le
        d_{\mathrm{TV}}\!\left(
        P_X^\ast\times \widehat P_{Y\mid X},
        P_X^\ast\times P_{Y\mid X}^\ast
        \right).
    \end{equation}
    Combining \eqref{eq:auxiliary_coverage_estimated} and \eqref{eq:tv_transfer_estimated} gives
    \begin{equation}
        \label{eq:conditional_estimated_coverage}
        \begin{split}
            \P\left(Y_{n+1}^\ast\in \hat{\cC}(X_{n+1})\right) 
            & \ge 1-\alpha - d_{\mathrm{TV}}\!\left(P_X^\ast\times \widehat P_{Y\mid X}, P_X^\ast\times P_{Y\mid X}^\ast\right) \\
            & \ge 1 - \alpha - \int d_{\mathrm{TV}}\!\left(\widehat P_{Y\mid X}(\cdot\mid x), P_{Y\mid X}^\ast(\cdot\mid x)\right) \,dP_X^\ast(x) \\
            & = 1 - \alpha - \E_{X \sim P^\ast_X}\left[d_{\mathrm{TV}}\!\left(\widehat P_{Y\mid X}(\cdot\mid X), P_{Y\mid X}^\ast(\cdot\mid X)\right)\right]
        \end{split}
    \end{equation}
    where the second inequality is from Lemma 2 of \cite{hur2026inference}. Note that we have considered the conditional probability given the training split so far. Now, by taking expectation in \eqref{eq:conditional_estimated_coverage} over the randomness in the training split, we obtain
    \[
        \P\left(Y_{n+1}^\ast\in \hat{\cC}(X_{n+1})\right) \ge 1-\alpha - \E_{tr} \E_{X \sim P^\ast_X}\left[d_{\mathrm{TV}}\!\left(\widehat P_{Y\mid X}(\cdot\mid X), P_{Y\mid X}^\ast(\cdot\mid X)\right)\right].
    \]
    Note that we have started the proof by observing $\P(\hat{\sigma}_o^\ast(Y_{n + 1}^\ast) \in \hat{\cC}(X_{n + 1})) \ge \P(Y_{n + 1}^\ast \in \hat{\cC}(X_{n + 1}))$. In fact, we have $\P(\hat{\sigma}_o^\ast(Y_{n + 1}^\ast) \in \hat{\cC}(X_{n + 1})) \ge \P(\sigma(Y_{n + 1}^\ast) \in \hat{\cC}(X_{n + 1}))$ for any permutation $\sigma$. Let $\tilde{Y}_{n + 1}^\ast = \sigma(Y_{n + 1}^\ast)$ so that the conditional distribution $\tilde{Y}_{n + 1}^\ast \mid X_{n + 1}$ follows $\sigma \circ P_{Y\mid X}^\ast(\cdot\mid X_{n + 1})$. Then, by repeating the proof above, one can deduce from \eqref{eq:tv_transfer_estimated} that
    \begin{equation*}
        \left|
        \P\left(\tilde{Y}_{n+1}\in \hat{\cC}(X_{n+1})\right) - \P\left(
        \tilde{Y}_{n+1}^\ast\in \hat{\cC}(X_{n+1})\right)\right|
        \le
        d_{\mathrm{TV}}\!\left(
        P_X^\ast\times \widehat P_{Y\mid X},
        P_X^\ast\times (\sigma \circ P_{Y\mid X}^\ast)
        \right).
    \end{equation*}
    Hence, we have 
    \begin{equation*}
        \P\left(\tilde{Y}_{n+1}^\ast\in \hat{\cC}(X_{n+1})\right) \ge 1-\alpha - \E_{tr} \E_{X \sim P^\ast_X}\left[d_{\mathrm{TV}}\!\left(\widehat P_{Y\mid X}(\cdot\mid X), \sigma \circ P_{Y\mid X}^\ast(\cdot\mid X)\right)\right].
    \end{equation*}
    Therefore, 
    \begin{equation*}
        \begin{split}
            \P(\hat{\sigma}_o^\ast(Y_{n + 1}^\ast) \in \hat{\cC}(X_{n + 1}))
            & \ge \max_{\sigma \in \cS_K} \P(\sigma(Y_{n + 1}^\ast) \in \hat{\cC}(X_{n + 1})) \\
            & \ge 1-\alpha - \min_{\sigma \in \cS_K} \E_{tr} \E_{X \sim P^\ast_X}\left[d_{\mathrm{TV}}\!\left(\widehat P_{Y\mid X}(\cdot\mid X), \sigma \circ P_{Y\mid X}^\ast(\cdot\mid X)\right)\right].
        \end{split}
    \end{equation*}
\end{proof}

\section{Additional Experiments}
We provide several additional experiments to illustrate properties of the proposed method.

\begin{figure}[b]
    \centering
    \includegraphics[width=\linewidth]{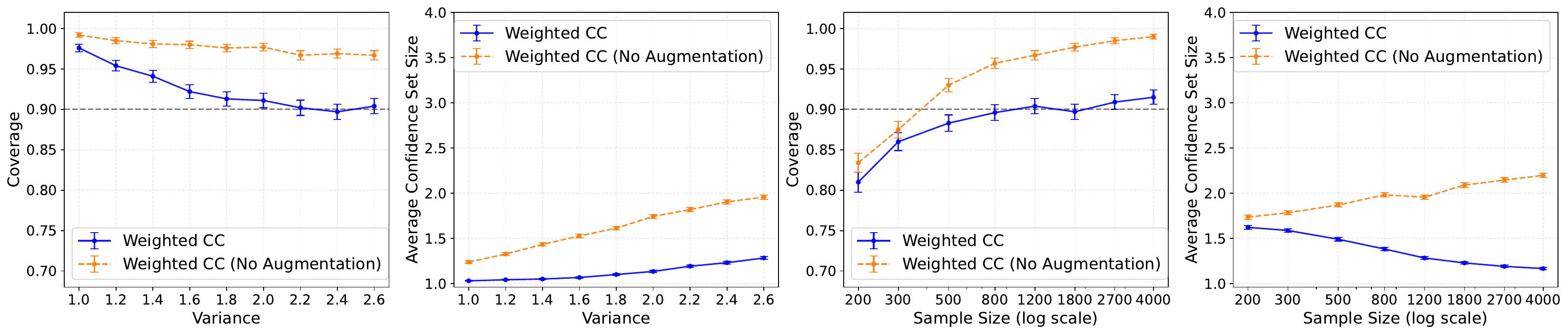}
    \vspace*{-10pt}
    \caption{Coverage and confidence set size for GMM simulations with $K = 5$ components for varying noise level $\sigma^2$ and sample size \(n\) in $\R^2$ as in the top row of Figure \ref{fig:GMM}. Algorithm \ref{alg:estimated_weighted_conformal_clustering} is compared with the version without the augmentation.}
    \label{fig:GMM_ablation}
\end{figure}

\paragraph{Ablation} Our method (Algorithm \ref{alg:estimated_weighted_conformal_clustering}) introduces a new scheme based on weighted conformal inference with augmentation, where the weights \eqref{eq:estimated_augmented_weights} are rigorously derived from theory. The augmentation was the key to making the infeasible oracle method (Algorithm \ref{alg:oracle_weighted_conformal_clustering}) tractable. However, one may wonder if there is another way to design weights without the augmentation, questioning whether the augmentation is necessary for our weighted conformal method. Currently, we believe Algorithm \ref{alg:estimated_weighted_conformal_clustering} with the weights \eqref{eq:estimated_augmented_weights} is the most sensible approach that is computationally tractable and theoretically guaranteed. To demonstrate that the augmentation is a necessary component of our method, we compare our method to a version of Algorithm \ref{alg:estimated_weighted_conformal_clustering} that uses the weights without the augmentation. In other words, we replace the denominator of the calibration weight $\hat{w}_i$ in \eqref{eq:estimated_augmented_weights} with ${\left[\mathcal A_{\cD_{ca}}(X_i)\right]_{Y_i}}$, namely, clustering with the calibration points only. Then, for $\hat{w}(y)$ in \eqref{eq:estimated_augmented_weights}, since we cannot define the denominator, we replace it with the one-hot assignment of $y$ based on calibration clustering, which is possible if calibration clustering is generalizable. We test this using the 2D GMM case in Figure \ref{fig:GMM} (top row). Figure \ref{fig:GMM_ablation} compares the weighted CC (blue, same as Figure \ref{fig:GMM}) to the method above without the augmentation (orange, dashed). We clearly see that without the augmentation, we end up with uninformative sets that are way larger than necessary, even when the sample size is sufficiently large. This demonstrates that the augmentation is an indispensable part of our algorithm.

\paragraph{Computational Cost}
Weighted conformal prediction is inherently more computationally expensive than split conformal prediction, as it requires refitting the model and recomputing the weights for each test point. In Tables~\ref{tab:runtime} and~\ref{tab:memory}, we compare the computational costs of Algorithm~\ref{alg:estimated_weighted_conformal_clustering} and Split-CC~\cite{hur2026inference} in terms of wall-clock time and peak memory usage, respectively. We consider the 2D Gaussian mixture model experiment (top row of Figure~\ref{fig:GMM}, varying sample size $n$), where both methods are evaluated on the same task of constructing prediction sets for $m=500$ test points at a nominal coverage level of $1-\alpha=0.9$. All computations are performed on a single core of an Apple M2 processor with 16\,GB of memory. The reported results are means (with standard deviations) across five random seeds, with peak memory usage measured using Python's \texttt{tracemalloc} module.

\begin{table}
\centering
\caption{Wall-clock time (in seconds) for computing prediction sets of the same $m = 500$ test points; mean (standard deviation) over 5 data seeds, single core.}
\label{tab:runtime}
\begin{tabular}{rccc}
\toprule
$n$ & Training (shared) & Split CC & Weighted CC \\
\midrule
200  & 0.0200 (0.0097) & 0.0180 (0.0061) & 8.37 (2.40) \\
300  & 0.0142 (0.0025) & 0.0143 (0.0029) & 7.77 (1.10) \\
500  & 0.0114 (0.0026) & 0.0120 (0.0016) & 6.42 (1.40) \\
800  & 0.0144 (0.0068) & 0.0144 (0.0004) & 6.56 (0.22) \\
1200 & 0.0128 (0.0011) & 0.0177 (0.0016) & 8.61 (0.77) \\
1800 & 0.0147 (0.0013) & 0.0219 (0.0017) & 10.4 (0.51) \\
2700 & 0.0196 (0.0015) & 0.0309 (0.0020) & 14.5 (0.50) \\
4000 & 0.0259 (0.0019) & 0.0461 (0.0034) & 22.0 (0.52) \\
\bottomrule
\end{tabular}
\end{table}

\begin{table}
\centering
\caption{Peak memory (in MB) for computing prediction sets of the same $m = 500$ test points; one data seed, single core, traced with Python's \texttt{tracemalloc}.}
\label{tab:memory}
\begin{tabular}{rccc}
\toprule
$n$ & Training (shared) & Split CC & Weighted CC \\
\midrule
200  & 0.06 & 0.21 & 0.49 \\
300  & 0.05 & 0.23 & 0.40 \\
500  & 0.07 & 0.27 & 0.59 \\
800  & 0.11 & 0.43 & 0.66 \\
1200 & 0.16 & 0.67 & 1.01 \\
1800 & 0.24 & 1.20 & 1.61 \\
2700 & 0.36 & 2.11 & 2.37 \\
4000 & 0.52 & 3.74 & 4.17 \\
\bottomrule
\end{tabular}
\end{table}

\paragraph{Misspecified Number of Clusters}
Currently, our framework (and that of \cite{hur2026inference}) focuses on the case where we know $K$, the number of true clusters. However, $K$ is unknown and needs to be estimated in practice. Given the early stage of this line of research on label uncertainty quantification, estimating $K$, which is a well-known challenge in the literature, is out of scope for this paper. That said, we provide a simple study illustrating the behavior of our algorithm with misspecified $K$ and discuss its potential usage for diagnosing the clustering fit. Figure \ref{fig:moon_misspecified} shows the heatmaps for the 2D moon example in Figure \ref{fig:2D_non_linear}. The middle panel corresponds to the second panel of the top row in Figure \ref{fig:2D_non_linear}, and the other two panels are produced by the same procedure with $K = 2, 4$ instead of the true $K = 3$, where the points are shown with different markers indicating their true clusters. For $K = 2$, notice that we have learned one cluster (left, orange region) correctly, while the other cluster combines two of the three true clusters. For $K = 4$, we see that there is a region (purple) that passes through the two clusters (right and bottom), leading to the undesirable clustering fit.

\begin{figure}
    \centering
    \includegraphics[width=\linewidth]{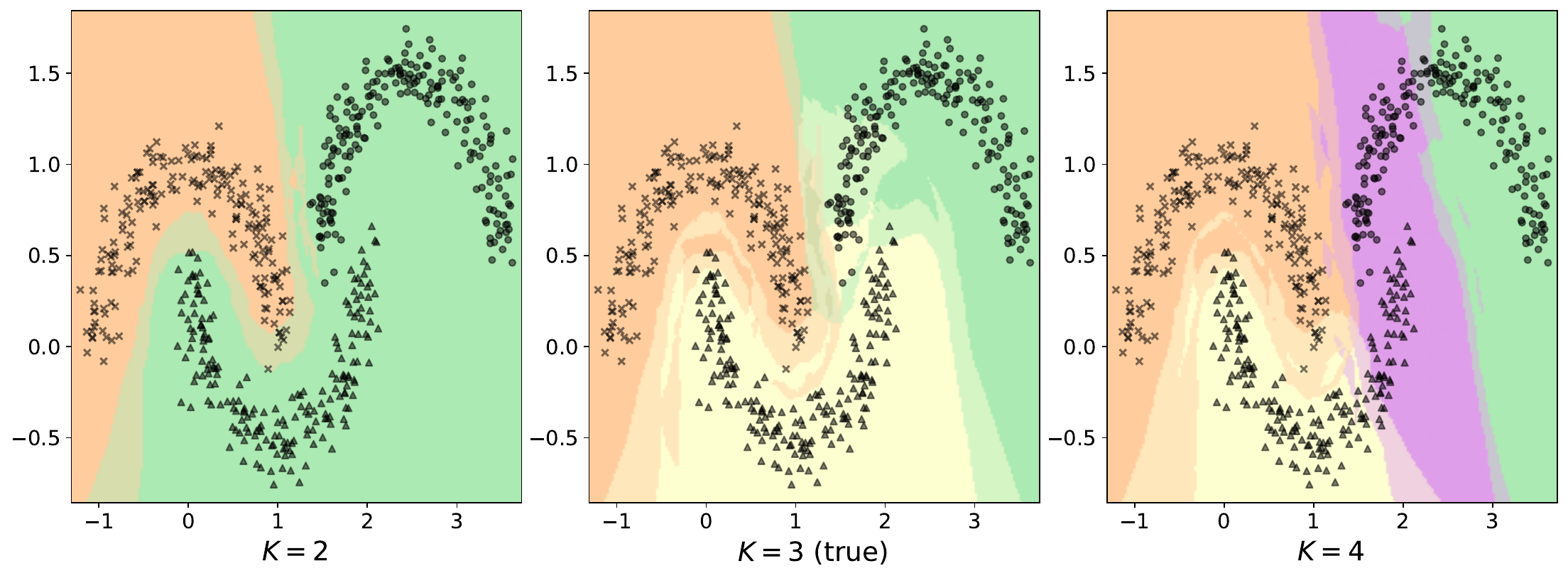}
    \vspace*{-10pt}
    \caption{Confidence set heatmaps for the nonlinear clustering simulation (same as the top row of Figure \ref{fig:2D_non_linear}) with $K = 2, 3, 4$.}
    \label{fig:moon_misspecified}
\end{figure}

\end{document}